\documentclass[11pt,centertags,reqno]{amsart}

\usepackage{epstopdf}
\usepackage{subfigure}
\usepackage{url}

\usepackage{color}

\newtheorem{theorem}{Theorem}[section]
\newtheorem{lemma}[theorem]{Lemma}
\newtheorem{proposition}[theorem]{Proposition}

\theoremstyle{remark}
\newtheorem{definition}[theorem]{Definition}

\newtheorem*{notation*}{Notations} 
\newtheorem{assumption}[theorem]{Assumption}
\newtheorem{remark}[theorem]{Remark}

\def \cF{{\mathcal F}}

\def \cD{{\mathcal{D}}}

\def \R{\mathbb{R}}

\def\half{\frac{1}{2}}
\def\D{\text{d}}
\def\hp{H_+}
\def\Oo{\mathcal{O}}
\def\NN{\mathbb{N}}
\def\NNt{\widetilde{\NN }}
\def \E{\mathrm{e}}
\def \I{\mathtt{i}}

\newcommand{\VIX}{\mathrm{VIX}}
\newcommand{\dd}{\boldsymbol{\mathrm{d}}}
\newcommand{\err}{\boldsymbol{\varepsilon}}
\usepackage{mathtools}
\mathtoolsset{showonlyrefs=true}
\newcommand{\Hf}{\mathfrak{H}}
\newcommand{\Hft}{\widehat{\mathfrak{H}}}

\begin{document}

 \author[O.~Bonesini]{Ofelia Bonesini}
\address{Ofelia Bonesini, Department of Mathematics, University of
Padova}\email{obonesin@ic.ac.uk}

\author[G.~Callegaro]{Giorgia Callegaro}
\address{Giorgia Callegaro, Department of Mathematics, University of
Padova, Via Trieste 63, 35121 Padova, Italy.}\email{gcallega@math.unipd.it}

\author[A.~Jacquier]{Antoine Jacquier}
\address{Antoine Jacquier, Department of Mathematics, Imperial College London, London SW7 1NE, UK.}\email{a.jacquier@imperial.ac.uk}

\title[]{Functional quantization of rough volatility and applications to volatility derivatives}

\date{\today}

\begin{abstract}
		We develop a product functional quantization of rough vo\-la\-ti\-li\-ty.
		Since the optimal quantizers can be computed offline, this new technique, 
		built on the insightful works by Luschgy and Pag\`es~\cite{Pages29, Sharp, Pages},
		becomes a strong competitor in the new arena of numerical tools for rough volatility.
		We concentrate our numerical analysis on the pricing of options on the VIX and realized variance in the rough Bergomi model~\cite{Bayer}
		and compare our results to other benchmarks recently suggested.
\end{abstract}
	
\maketitle


%

%
%
%
%
%

\section{Introduction}\label{intro}
Gatheral, Jaisson and Rosenbaum~\cite{Vol_is_r} recently introduced a new framework for financial modelling. 
To be precise --- according to the reference website \url{https://sites.google.com/site/roughvol/home} --- almost twenty-four hundred days have passed since instantaneous volatility was shown to have a rough nature, 
in the sense that its sample paths are $\alpha$-H\"{o}lder-continuous with $\alpha < \frac{1}{2}$.
Many studies, both empirical~\cite{Bennedsen, Vol_is_r?,Fukasawa} and theoretical~\cite{Fukasawa_T, Alos_T},
have confirmed this, showing that these so-called rough volatility models are a more accurate fit to the implied volatility surface and to estimate historical volatility time series.

On equity markets, the quality of a model is usually measured by its ability to calibrate not only to the SPX implied volatility but also VIX Futures and the VIX implied volatility.
The market standard models had so far been Markovian, 
in particular the double mean-reverting process~\cite{Gatheral_Mark_Model, Huh}, Bergomi's model~\cite{Bergomi} and, to some extent, jump models~\cite{Carr, Kokholm}.
However, they each suffer from several drawbacks, which the new generation of rough volatility models seems to overcome.
For VIX Futures pricing, the rough version of Bergomi's model was thoroughly investigated in~\cite{Muguruza}, showing accurate results.
Nothing comes for free though and the new challenges set by rough volatility models lie on the numerical side,
as new tools are needed to develop fast and accurate numerical techniques.
Since classical simulation tools for fractional Brownian motions are too slow for realistic purposes, new schemes have been proposed to speed it up, 
among which the Monte Carlo hybrid scheme~\cite{Bennedsen, McCrickerd}, 
a tree formulation~\cite{RoughTrees}, quasi Monte-Carlo methods~\cite{Bayer_MC}
and Markovian approximations~\cite{AbiJaber, Chen}. 

We suggest here a new approach, based on product functional quantization~\cite{Pages}.
Quantization was originally conceived as a discretization technique to approximate a continuous signal by a discrete one~\cite{Sheppard}, 
later developed at  Bell Laboratory in the 1950s for signal transmission~\cite{Gersho}. 
It was however only in the 1990s that its power to compute (conditional) expectations of functionals of random variables~\cite{GrafLus}
was fully understood.
Given an $\R^d$-valued random vector on some probability space, 
optimal vector quantization investigates how to select an $\R^d$-valued random vector~$\widehat{X}$, 
supported on at most~$N$ elements, that best approximates~$X$ according to a given criterion
(such as the $L^r$-distance, $r\geq 1$).
Functional quantization is the infinite-dimensional version, 
approximating a stochastic process with a random vector taking a finite number of values in the space of trajectories for the original process.
It has been investigated precisely~\cite{Pages29, Pages} in the case of Brownian diffusions,
in particular for financial applications~\cite{PP2005}.
However, optimal functional quantizers are in general hard to compute numerically and instead product functional quantizers provide a rate-optimal (so, in principle, sub-optimal) alternative often admitting closed-form expressions~\cite{Sharp, PP2005}.
\\

In Section~\ref{Sect_RL} we briefly review important properties of \emph{Gaussian Volterra processes},
 displaying a series expansion representation,
and paying special attention to the \emph{Riemann-Liouville} case in Section~\ref{Sect_RL_1}.
This expansion yields, in Section~\ref{section_quant}, a product functional quantization of the processes, 
that shows an $L^2$-error of order $\log(N)^{-H}$, with~$N$ the number of paths and~$H$ a regularity index.
We then show, in Section~\ref{Sect_staz}, that these functional quantizers, although sub-optimal, are stationary.
We specialise our setup to the generalized rough Bergomi model in Section~\ref{VIX_F} 
and show how product functional quantization applies to the pricing of VIX Futures and VIX options, proving in particular precise rates of convergence.
Finally, Section~\ref{Sect_Num_res} provides a numerical confirmation of 
the quality of our approximations for VIX Futures  and Call Options on the VIX in the rough Bergomi model,
benchmarked against other existing schemes.
In this Section, we also discuss how product functional quantization of the Riemann-Liouville process itself can be exploited to price options on realized variance.

\begin{notation*}
	We set~$\NN$ as the set of strictly positive natural numbers. We denote by~$\mathcal C[0,1]$ the space of real-valued continuous functions over $[0,1]$ and by $L^2[0,1]$ 
	 the Hilbert space of real-valued square integrable functions on $[0,1]$, with inner product   $\langle f, g \rangle_{L^2[0,1]}:= \int_0^1 f(t)g(t) dt$,
inducing the norm $\| f \|_{L^2[0,1]} := ( \int_0^1 |f(t)|^2 dt)^{1/2}$, for each $ f, g \in L^2[0,1]$.
$L^2(\mathbb P)$ denotes the space of square integrable (with respect to $\mathbb P$) random variables.
\end{notation*}

\section{Gaussian Volterra processes on $\mathbb R_+$}\label{Sect_RL}

For clarity, we restrict ourselves to the time interval $[0,1]$.
Let $\{W_{t}\}_{t\in[0,1]}$ be a standard Brownian motion on a filtered probability space $(\Omega,\cF,\{\cF_t\}_{t \in[0,1]},\mathbb{P})$, with $\{\cF_t\}_{t \in[0,1]}$ its natural filtration.
On this probability space we introduce the Volterra process
\begin{equation}\label{GV}
	Z_t
	:= \int_0^t K(t-s)dW_s, \qquad t\in [0,1],
\end{equation}
and we consider the following assumptions for the kernel~$K$:
\begin{assumption}\label{Assumpt_Bennedsen}\ 
There exist $\alpha \in \left( -\half, \half\right)\setminus \{0\}$ and $L:(0,1] \to (0, \infty)$ continuously differentiable, slowly varying at $0$, that is, for any $t>0$, $\lim_{x \downarrow 0} \frac{L(t x)}{L(x)} = 1$, and bounded away from $0$ function with
$|L'(x)| \leq C(1+x^{-1})$, for $x \in (0,1]$, for some $C>0$, such that
$$
K(x)= x^{\alpha}L(x), \qquad x \in (0,1].
$$
\end{assumption}

This implies in particular that $K \in L^2[0,1]$, so that the stochastic integral~\eqref{GV} is well defined.
The Gamma kernel, with $K(u)= e^{-\beta u} u^{\alpha}$, for $\beta > 0$
and $\alpha \in (-\frac{1}{2}, \frac{1}{2})\setminus \{0\}$,  is a classical example satisfying Assumption~\ref{Assumpt_Bennedsen}.  
Straightforward computations show that the covariance function of~$Z$ 
reads
\begin{equation}\label{cov}
	R_{Z}(s,t)=\int_0^{t \wedge s}K(t-u)K(s-u)du, \quad s, t \in [0,1].
\end{equation}

Under Assumption~\ref{Assumpt_Bennedsen}, $Z$ is a Gaussian process admitting a version which is 
$\varepsilon$-H\"{o}lder continuous for any $\varepsilon < \frac{1}{2} + \alpha=H$ 
and hence also admits a continuous version~\cite[Proposition 2.11]{Bennedsen}.


\subsection{Series expansion}\label{RL:series_repr}
We introduce a series expansion representation for the centered Gaussian process~$Z$ in~\eqref{GV},
which will be key to develop its functional quantization. 
Inspired by~\cite{Sharp}, introduce the stochastic process
\begin{equation}\label{series}
	Y_t{:=}\sum_{n\geq 1} \mathcal{K}[\psi_n](t) \xi_n, \qquad t \in [0,1],
\end{equation}
where $\{\xi_n\}_{n\geq1}$ is a sequence of i.i.d. standard Gaussian random variables,
$\{\psi_n\}_{n\geq1}$ denotes the orthonormal basis of $L^2[0,1]$:
\begin{equation}\label{eq_psi_n}
	\psi_n(t)=\sqrt{2}\cos\left(\frac{t}{\sqrt{\lambda_n}}\right),
	\quad \text{ with }\lambda_n=\frac{4}{(2n-1)^2 \pi^2},
\end{equation}
and the operator $\mathcal{K} : {L}^2[0,1] \to  \mathcal{C}[0,1]$ is defined for $f \in {L}^2[0,1]$ as 
\begin{equation}\label{def_ker_KH}
	\mathcal{K}[f](t)
	:=\int_0^t K(t-s)f(s)ds,
\qquad\text{for all }t \in [0,1].
\end{equation}

\begin{remark}\label{gau}
The stochastic process $Y$ in~\eqref{series} is defined as a weighted sum of independent centered Gaussian variables, so for every $t \in [0,1]$ the random variable $Y_t$ is a centered Gaussian random variable and the whole process $Y$ is Gaussian with zero mean. 
\end{remark}


We set the following assumptions on the functions $\{\mathcal{K}[\psi_n]\}_{n \in \NN}$:
\begin{assumption}\label{AssumptionAB}
There exists $H \in (0,\half)$ such that
\begin{itemize}
\item[\textbf{(A)}] there is a constant $C_1 > 0$ for which, for any $ n \geq 1$, $ \mathcal{K}[\psi_n]$ is $(H+\frac{1}{2})$-H\"older continuous, with 
$$
\sup_{s,t \in [0,1], s\neq t}\frac{| \mathcal{K}[\psi_n] (t)- \mathcal{K}[\psi_n] (s) |}{|t-s|^{H+\frac{1}{2}}} \leq C_1 n;
$$
		\item[\textbf{(B)}] 
		there exists a  constant $C_2 > 0$ such that
		$$
		\sup_{t \in [0,1]}|\mathcal{K}[\psi_n](t)| \leq C_2 n^{-(H+\frac{1}{2})}, \quad \text{ for all } n \geq 1.
		$$
	\end{itemize}
\end{assumption}

Notice that under these assumptions, the series~\eqref{series} converges both almost surely and in $L^2(\mathbb{P})$ for each $t \in [0,1]$ by Khintchine-Kolmogorov Convergence Theorem~\cite[Theorem 1, Section 5.1]{Chow}.

It is natural to wonder whether Assumption~\ref{Assumpt_Bennedsen} implies
Assumption~\ref{AssumptionAB} given the basis functions~\eqref{eq_psi_n}. This is far from trivial in our general setup and we provide examples and justifications later on for models of interest. Similar considerations with slightly different conditions can be found in~\cite{Sharp}.
We now focus on the variance-covariance structure of the Gaussian process $Y$.
\begin{lemma}\label{lem:covY}
	For any $s,t \in [0,1]$, the covariance function of~$Y$ is given by
	$$
	R_{Y}(s,t):= \mathbb{E}[Y_s Y_t] = \int_0^{t \wedge s} K(t-u)K(s-u) du. 
	$$
\end{lemma}
\begin{proof}
	Exploiting the definition of $Y$ in~\eqref{series}, the definition of $\mathcal{K}$ in~\eqref{def_ker_KH} and the fact that the random variable $\xi_n$'s are i.i.d. standard Normal, 
	we obtain
	\begin{eqnarray*}
		R_{Y}(s,t)&=&\mathbb{E}[Y_s Y_t]
		=\mathbb{E}\Bigg[\Big(\sum_{n\geq 1} \mathcal{K}[\psi_n](s) \xi_n\Big)\Big(\sum_{m\geq 1} \mathcal{K}[\psi_m](t) \xi_m\Big)\Bigg]
		=\sum_{n\geq 1} \mathcal{K}[\psi_n](s)\mathcal{K}[\psi_n](t)\\&
		= & \sum_{n \geq 1} \Big( \int_0^1 K(s-u) \bold 1_{[0,s]}(u) \psi_n(u) du  \int_0^1 K(t-r) \bold 1_{[0,t]}(r) \psi_n(r) dr \Big)\\&
		= & \sum_{n \geq 1}  \langle K(s-\cdot) \bold 1_{[0,s]}(\cdot), \psi_n \rangle_{L^2[0,1]}  \cdot \langle K(t-\cdot)\bold 1_{[0,t]}(\cdot), \psi_n \rangle_{L^2[0,1]}\\&
		= & \sum_{n \geq 1}  \Big\langle K(t-\cdot) \bold 1_{[0,t]}(\cdot),\langle K(s-\cdot)\bold 1_{[0,s]}(\cdot), \psi_n \rangle_{L^2[0,1]}  \psi_n \Big\rangle_{L^2[0,1]}\\&
		= & \Big\langle K(t-\cdot) \bold 1_{[0,t]}(\cdot),\sum_{n \geq 1} \langle K(s-\cdot) \bold 1_{[0,s]}(\cdot), \psi_n \rangle_{L^2[0,1]}  \psi_n \Big\rangle_{L^2[0,1]}\\&
		= & \langle K(t-\cdot) \bold 1_{[0,t]}(\cdot) , K(s-\cdot) \bold 1_{[0,s]}(\cdot) \rangle_{L^2[0,1]}\\&
		= & \int_0^1  K(s-u) \bold 1_{[0,s]}(u) K(t-u) \bold 1_{[0,t]}(u) du 
		= \int_0^{t \wedge s} K(t-u) K(s-u) du. 
	\end{eqnarray*}
\end{proof}

\begin{remark}\label{Cont_v_Y}
	Notice that the centered Gaussian stochastic process $Y$ admits a continuous version, too. 
	Indeed, we have shown that $Y$ has the same mean and covariance function as $Z$ and, consequently, that the increments of the two processes share the same distribution. 
	Thus,~\cite[Proposition 2.11]{Bennedsen} applies to $Y$ as well, yielding that the process admits a continuous version.
	This last key property of~$Y$ can be alternatively proved directly as done in Appendix~\ref{App_P_Cont_v_Y}.
\end{remark}

Lemma~\ref{lem:covY} implies that $\mathbb{E}[Y_s Y_t]=\mathbb{E}[Z_s Z_t],$ for all $ s,t \in [0,1]$. 
Both~$Z$ and~$Y$ are continuous, centered, Gaussian with the same covariance structure,
so from now on we will work with~$Y$, using
\begin{equation}\label{eq:Zseries}
Z	= \sum_{n \geq 1} \mathcal{K}[\psi_n] \xi_n,\quad \mathbb{P}\text{-a.s.}
\end{equation}

\subsection{The Riemann - Liouville case}\label{Sect_RL_1}

For $K(u) = u^{H-\frac{1}{2}} $, with $H \in (0, \frac{1}{2})$, 
the process~\eqref{GV} takes the form
\begin{equation}\label{RL}
	Z^H_t
	:= \int_0^t (t-s)^{H-\frac{1}{2}}dW_s
	, \qquad t\in [0,1],
\end{equation}
where we add the superscript~$H$ to emphasise its importance. 
It is called a \emph{Riemann-Liouville process} (henceforth RL)
(also known as \emph{Type II fractional Brownian motion} or \emph{ L\'evy fractional Brownian motion}), 
as it is obtained by applying the Riemann-Liouville fractional operator to the standard Brownian motion, and is an example of a Volterra process.
This process enjoys properties similar to those of the fractional Brownian motion
(fBM), in particular being $H$-self-similar and centered Gaussian.
However, contrary to the fractional Brownian motion, its increments are not stationary.
For a more detailed comparison between the fBM and $Z^H$ we refer to~\cite[Theorem 5.1]{Picard}.
In the RL case, the covariance function $R_{Z^H}(\cdot,\cdot)$ is available~\cite[Proposition~2.1]{Jacquier} explicitly as
\begin{displaymath}
	R_{Z^H}(s,t)= \frac{1}{H+ \frac{1}{2}}(s \land t)^{H+\frac{1}{2}}(s \lor t)^{H-\frac{1}{2}} \ {}_2 F_1 \left( 1,\frac{1}{2}-H;  2H+1;  \frac{s \land t}{s \vee t} \right), \quad s, t \in [0,1],
\end{displaymath}
where ${}_2 F_1(a,b;c;z)$ denotes the Gauss hypergeometric function ~\cite[Chapter~5, Section~9]{Olver}.
More generally,~\cite[Chapter~5, Section~11]{Olver}, the generalized Hypergeometric functions ${}_p F_q(z)$ are defined as 
\begin{equation} \label{eq:Fpq}
	{}_p F_q(z) 
	= {}_p F_q( a_1, a_2, \dots, a_p;c_1,c_2, \dots, c_q;z) 
	:= \sum_{k=0}^\infty \frac{ {(a_1)}_k {(a_2)}_k \cdots {(a_p)}_k }{ {(c_1)}_k {(c_2)}_k \cdots {(c_q)}_k} \frac{z}{k !},
\end{equation}
with the Pochammer's notation ${(a)}_0:=1$ and ${(a)}_k := a (a+1) (a+2) \cdots (a+k-1)$, for $k\geq 1$, where none of the~$c_k$ are negative integers or zero. 
For $p \le q$ the series~\eqref{eq:Fpq} converges for all~$z$ and when $p=q+1$ convergence holds for $|z| < 1$ 
and the function is defined outside this disk by analytic continuation. 
Finally, when $p > q+1$ the series diverges for nonzero~$z$ unless one of the $a_k$'s is zero or a negative integer.

Regarding the series representation~\eqref{series}, we have, for $t \in [0,1]$ and $ n \geq 1$,
\begin{align}\label{eq_rep_KH}
	\mathcal{K}_H[\psi_n](t)
	 : & =  \sqrt{2} \int_0^t (t-s)^{H-\frac{1}{2}} \cos\Big( \frac{s}{\sqrt{\lambda_n}} \Big) ds\\
	& =  \frac{2\sqrt{2}}{1+2H} \ t^{H+\frac{1}{2}} \ {}_1F_2\left(1;\frac{3}{4}+\frac{H}{2},\frac{5}{4}+\frac{H}{2};-\frac{t^2}{4 \lambda_n} \right).
\end{align}

Assumption~\ref{AssumptionAB} holds in the RL case here using~\cite[Lemma 4]{Sharp}
(identifying $\mathcal{K}_H[\psi_n]$ to~$f_n$ from~\cite[Equation (3.7)]{Sharp}).
Assumption~\ref{AssumptionAB} \textbf{(B)} implies that,
for all $t \in [0,1]$,
\[
	\sum_{n\geq 1} \mathcal{K}_H[\psi_n](t)^2
	\leq \sum_{n \geq 1} \left(\sup_{t \in [0,1]}|\mathcal{K}_H[\psi_n](t)| \right)^2 
	\leq C_2^2 \sum_{n\geq 1 } \frac{1}{n^{1+2H}}
	< \infty,
\]
and therefore the series~\eqref{series} converges both almost surely and in $L^2(\mathbb{P})$ for each $t \in [0,1]$ by Khintchine-Kolmogorov Convergence Theorem~\cite[Theorem 1, Section 5.1]{Chow}.

\begin{remark}
The expansion~\eqref{series} is in general not a Karhunen-Lo\`eve decomposition~\cite[Section 4.1.1]{PP2005}. 
In the RL case, it can be numerically checked that the basis $\{\mathcal{K}_H[\psi_n]\}_{n \in \NN}$ is not orthogonal in $L^2[0,1]$ and does not correspond to eigenvectors for the covariance operator of the Riemann-Liouville process. 
In his PhD Thesis~\cite{Corlay}, Corlay exploited  a numerical method to obtain approximations of the first terms in the K-L expansion of processes for which an explicit form is not available.
\end{remark}

\section{Functional quantization and error estimation}\label{section_quant}
Optimal (quadratic) vector quantization was conceived to approximate a square integrable random vector $X: (\Omega, \mathcal F, \mathbb P) \rightarrow \mathbb R^d$ 
by another  one~$\widehat X$, taking at most a finite number~$N$ of values, 
on a grid $\Gamma^N := \{ x_1^N, x_2^N, \dots, x_N^N \}$, with $x_i^N \in \mathbb R^d, i=1,\dots, N$. 
The quantization of $X$ is defined as $\widehat X:= \mathrm{Proj}_{\Gamma^N}(X) $, where $\mathrm{Proj}_{\Gamma^N}: \mathbb R^d \rightarrow \Gamma^N$ denotes the nearest neighbour projection. 
Of course the choice of the $N$-quantizer $\Gamma^N$ is based on a given optimality criterion: in most cases $\Gamma^N$ minimizes the distance $\mathbb{E}[ | X - \widehat X |^2]^{1/2}$. 
We recall basic results for one-dimensional standard Gaussian, 
which shall be needed later, and refer to~\cite{GrafLus} for a comprehensive introduction to quantization.

\begin{definition}\label{def_quant_normal_rv}
	Let $\xi$ be a one-dimensional standard Gaussian on a probability space $(\Omega,\cF,\mathbb{P})$. 
For each $n \in \NN$, we define the \emph{optimal quadratic $n$-quantization of $\xi$}  as the random variable $\widehat{\xi}^n:= \mathrm{Proj}_{\Gamma^n}(\xi) = \sum_{i=1}^n x_i^n 1_{C_i(\Gamma^n)}(\xi)$, where $\Gamma^n = \{x_1^n, \dots,x_n^n\}$ is the unique optimal quadratic $n$-quantizer of $\xi$, namely the unique solution to the minimization problem
	\[
		\min_{\Gamma^n \subset \R, \mathrm{Card}(\Gamma^n)=n} \mathbb{E}[|\xi - \mathrm{Proj}_{\Gamma^n}(\xi)|^2],
	\]
	and $\{C_i(\Gamma^n)\}_{i \in \{1,\dots,n\}}$ is a Voronoi partition of $\R$, that is a Borel partition of $\R$ that satisfies
	\[
	C_i(\Gamma^n) 
	\subset \left\{ y  \in \R: |y-x_i^n| = \min_{1 \leq j \leq n}|y-x_j^n|\right\} 
	\subset \overline{C}_i(\Gamma^n),
	\]
where the right-hand side denotes the closure of the set in~$\R$.
\end{definition}

The unique optimal quadratic $n$-quantizer $\Gamma^{n} = \{ x_1^{n}, \dots, x_{n}^{n}  \}$ and the corresponding quadratic error are available online, at \url{http://www.quantize.maths-fi.com/gaussian_database} for $n\in \{1, \dots, 5999 \}$.

Given a stochastic process, viewed as a random vector taking values in its trajectories space, such as $L^2[0,1]$, functional quantization does the analogue to vector quantization in an infinite-dimensional setting, approximating the process with a finite number of trajectories. 
In this section, we focus on product functional quantization of the centered Gaussian process~$Z$ 
from~\eqref{GV} of order~$N$ (see~\cite[Section 7.4]{Pages} for a general introduction to product functional quantization). 
Recall that we are working with the continuous version of~$Z$ in the series~\eqref{eq:Zseries}.
For any $m, N \in \NN$, we introduce the following set, which will be of key importance all throughout the paper:
\begin{equation}\label{eq:SetD}
\cD_m^N := \left\{ \dd \in \NN^m : \prod_{i=1}^m d(i) \leq N\right\}.
\end{equation}
\begin{definition}\label{def_quant_Z_H}
	A  \emph{product functional quantization of $Z$ of order $N$} is defined as
	\begin{equation}\label{p-quant}
		\widehat{Z}^{\dd}_t 
		:= \sum_{n=1}^{m} \mathcal{K}[\psi_n](t) \widehat{\xi}_n^{d(n)}, \qquad t \in [0,1],
	\end{equation}
	where $\dd\in\cD_m^N$, for some $m \in \mathbb{N}$, and for every $ n \in \{1,\dots,m\},$ $\widehat{\xi}_n^{d(n)}$ is the (unique) optimal quadratic quantization of the standard Gaussian random variable $\xi_n$ of order $d(n)$, according to Definition~\ref{def_quant_normal_rv}. 
\end{definition}

\begin{remark}\label{remark_on_quant}
The condition $\prod_{i=1}^m d(i) \leq N$ in Equation~\eqref{eq:SetD} motivates the wording `product' functional quantization. Clearly, the optimality of the quantizer also depends on the choice of~$m$ and~$\dd$, for which we refer to  Proposition~\ref{prop:choice_m_d} and Section ~\ref{sect_m_d_in_pract}.
\end{remark}

Before proceeding, we need to make precise the explicit form for the product functional quantizer of the stochastic process $Z$:

\begin{definition}\label{def_quantizer_Z}
	The \emph{product functional 
	$\dd$-quantizer of $Z$} is defined as 
	\begin{equation} 
		\chi_{\underline i}^{\dd}(t) 
		:= \sum_{n=1}^{m} \mathcal{K}[\psi_n](t) \ x_{i_n}^{d(n)}, \qquad t \in [0,1], \quad  \underline i = (i_1, \dots, i_m), 
	\end{equation}
	for $\dd\in\cD_m^N$ and $1\leq i_n \leq d(n)$ for each $ n=1,\dots,m.$
\end{definition}

\begin{remark}\label{re_quantizer_Z}
	Intuitively, the quantizer is chosen as a Cartesian product of grids of the one-dimensional standard Gaussian random variables. 
	So, we also immediately find the probability associated to every trajectory $\chi_{\underline i}^{\dd}$:
	for every $\underline i = (i_1, \dots, i_m) \in \prod_{n=1}^m \{ 1, \dots, d(n)  \}$,
	\begin{equation}
		\mathbb P(\widehat{Z}^{\dd}  			= \chi_{\underline i}^{\dd}) = \prod_{n=1}^m \mathbb 	P(\xi_n \in C_{i_n}(\Gamma^{d(n)}) ),
	\end{equation}
	where $C_j(\Gamma^{d(n)})$ is the $j$-th Voronoi cell relative to the $d(n)$-quantizer $\Gamma^{d(n)}$ in Definition~\ref{def_quant_normal_rv}.
\end{remark}

The following, proved in Appendix~\ref{App_P_prop:choice_m_d}, 
deals with the quantization error estimation and its minimization 
and provides hints to choose $(m, \dd)$.
A similar result on the error can be obtained applying~\cite[Theorem 2]{Sharp} to the first example provided in the reference. 
For completeness we preferred to prove the result in an autonomous way in order to further characterize the explicit expression of the rate optimal parameters. 
Indeed, we then compare these rate optimal parameters with the (numerically computed) optimal ones  in Section~\ref{sect_m_d_in_pract}.
The symbol $\lfloor \cdot \rfloor$ denotes the lower integer part.

\begin{proposition}\label{prop:choice_m_d}
Under Assumption~\ref{AssumptionAB}, 
for any $N \geq 1$, there exist $m^*(N) \in \NN$ and  $C>0$ such that
	\[
\mathbb{E} \left[ \left\| \widehat{Z}^{\dd^*_N} - Z \right\|^2_{L^2[0,1]} \right]^{\half} 
\leq C \log(N)^{-H},
	\]
where $\dd^*_N\in \cD_{m^*(N)}^{N}$ and with, for each $n=1, \dots, m^*(N)$,
	\begin{equation}
	d^*_N(n) 
	= \Big\lfloor  N^{\frac{1}{m^*(N)}} n^{-(H+\frac{1}{2})} \left(m^*(N)!\right)^{\frac{2H+1}{2m^*(N)}}  \Big\rfloor.
	\end{equation}
	Furthermore $ m^*(N)= \mathcal{O}(\log(N))$.
\end{proposition}

\begin{remark}\label{rem_choice_m_and_d}
	In the RL case, the trajectories of $\widehat{Z}^{H, \dd}$ are easily computable and they are used in the numerical implementations  to approximate the process $Z^H$.
In practice, the parameters~$m$ and $\dd=(d(1),\dots,d(m))$ are chosen as explained in Section~\ref{sect_m_d_in_pract}. 
\end{remark}

\subsection{Stationarity}\label{Sect_staz}
We now show that the quantizers we are using are stationary. 
The use of stationary quantizers is motivated by the fact that their expectation provides a lower bound for the expectation of convex functionals of the process (Remark~\ref{rem_staz_quant}) and they yield a lower (weak) error in cubature formulae~\cite[page 26]{Pages}.
We first recall the definition of stationarity for the quadratic quantizer of a random vector~\cite[Definition~1]{Pages}.
\begin{definition}
Let $X$ be an $\R^d$-valued random vector on $(\Omega, \cF, \mathbb{P})$. 
A quantizer~$\Gamma$ for~$X$ is \emph{stationary} if  the nearest neighbour projection $\widehat{X}^\Gamma = \mathrm{Proj}_\Gamma(X)$ satisfies 
	\begin{equation}\label{eq_stat_quant}
		\mathbb{E}\left[X | \widehat X^\Gamma\right] = \widehat X^\Gamma.
	\end{equation}
\end{definition}

\begin{remark}\label{rem_staz_quant}
	Taking expectation on both sides of~\eqref{eq_stat_quant} yields $
	\mathbb{E}[X]
	=\mathbb{E}[\mathbb{E}[X | \widehat X^\Gamma]] 
	= \mathbb{E}[\widehat X^\Gamma].
	$
	Furthermore, for any convex function $f: \R^d \to \R$, the identity above, 
the conditional Jensen's inequality and the tower property yield
	\[
	\mathbb{E}[f(\widehat X^\Gamma)]
	=\mathbb{E}[ f(\mathbb{E}[X | \widehat X^\Gamma])] 
	\leq \mathbb{E}[ \mathbb{E}[f(X) | \widehat X^\Gamma]] 
	= \mathbb{E}[f(X)].
	\]
\end{remark}

While an optimal quadratic quantizer of order~$N$ of a random vector is always stationary~\cite[Proposition 1(c)]{Pages}, the converse is not true in general. 
  We now present the corresponding definition for a stochastic process.
\begin{definition}
Let $\{X_{t}\}_{t \in [T_1,T_2]}$ be a stochastic process on $(\Omega, \cF,\{\cF_t\}_{t \in [T_1,T_2]}, \mathbb{P})$. We say that an $N$-quantizer $\Lambda^N := \{\lambda_1^N, \cdots, \lambda_N^N \} \subset L^2[T_1,T_2] $, inducing the quantization $\widehat{X}= \widehat X^{\Lambda^N}$, is \emph{stationary} if 
$\mathbb{E}[X_t | \widehat X_t] = \widehat X_t$, for all $t \in [T_1,T_2]$.
\end{definition}

\begin{remark}
To ease the notation, we omit the grid  $\Lambda^N$ in $ \widehat X^{\Lambda^N}$, while the dependence on the dimension~$N$ remains via the superscript $\dd\in\cD_m^N$ (recall \eqref{p-quant}). 
\end{remark}

As was stated in Section~\ref{RL:series_repr}, we are working with the continuous version of the Gaussian Volterra process $Z$ given by the series expansion~\eqref{eq:Zseries}.
This will ease the proof of stationarity below (for a similar result in the case of the Brownian motion~\cite[Proposition 2]{Pages}).
\begin{proposition}\label{Prop_staz_quant}
	The product functional quantizers inducing~$\widehat{Z}^{\dd}$ in~\eqref{p-quant} are stationary.
\end{proposition}

\begin{proof}
	For any $t \in [0,1]$, by linearity, we have the following chain of equalities:
$$
\mathbb{E}\left[Z_t |\{\widehat{\xi}_n^{d(n)}\}_{1 \le n \le m}\right] =\mathbb{E}\left[ \sum_{k\geq 1} \mathcal{K}[\psi_k](t) \xi_k  \Big|\{\widehat{\xi}_n^{d(n)}\}_{1 \le n \le m }\right]
=\sum_{k \ge 1} \mathcal{K}[\psi_k](t) \mathbb{E}\left[\xi_k \Big|\{\widehat{\xi}_n^{d(n)}\}_{1 \le n \le m }\right].
$$
	
Since the $\mathcal N(0,1)$-Gaussian $\xi_n$'s are i.i.d., by definition of  optimal quadratic quantizers (hence stationary), we have
$\mathbb{E}[\xi_k | \widehat{\xi}_i^{d(i)}] = \delta_{i k}\widehat{\xi}_i^{d(i)}$,
for all $i, k \in \{1,\dots,m\}$,
and therefore
	\[
		\mathbb{E}\left[\xi_k \Big|\{\widehat{\xi}_n^{d(n)}\}_{1 \le n \le m }\right]
		=\mathbb{E}\left[\xi_k \Big|  \widehat{\xi}_k^{d(k)}\right]
		= \widehat{\xi}_k^{d(k)},
		\text{ for all }k \in \{1,\dots,m\}.
	\]
	Thus, we obtain
$$
\mathbb{E}\left[Z_t \Big|\{\widehat{\xi}_n^{d(n)}\}_{1 \le n \le m}\right ]
		=\sum_{k \ge 1} \mathcal{K}[\psi_k](t) \widehat{\xi}_k^{d(k)}
		= \widehat{Z}^{\dd}_t.
$$
	Finally, exploiting the tower property and the fact that the $\sigma$-algebra generated by $\widehat{Z}^{\dd}_t$ is included in the $\sigma$-algebra generated by $\{\widehat{\xi}_n^{d(n)}\}_{n \in \{1, \dots,m\}}$ by Definition~\ref{def_quant_Z_H}, we obtain
	\begin{eqnarray*}
		\mathbb{E}\left[Z_t \Big|  \widehat{Z}^{\dd}_t \right]&
		= \mathbb{E}\left[\mathbb{E}\left[Z_t \Big|\{\widehat{\xi}_n^{d(n)}\}_{n \in \{1, \dots,m\}} \right]\Big|  \widehat{Z}^{\dd}_t \right]
		= \mathbb{E}\left[\widehat{Z}^{\dd}_t\Big|  \widehat{Z}^{\dd}_t \right]
		= \widehat{Z}^{\dd}_t,
	\end{eqnarray*}
	which concludes the proof.
\end{proof}

\section{Application to VIX derivatives in rough Bergomi} \label{VIX_F}
We now specialize the setup above to the case of rough volatility models.
These models are extensions of classical stochastic volatility models, introduced to better reproduce the market implied volatility surface.
The  volatility process is stochastic and driven by a rough process,
by which we mean a process whose trajectories are $H$-H\"older continuous with $H\in (0, \half)$. 
The empirical study~\cite{Vol_is_r}  was the first to suggest such a rough behaviour for the volatility, and ignited tremendous interest in the topic.
The website
\url{https://sites.google.com/site/roughvol/home}
contains an exhaustive and up-to-date review of the literature on rough volatility.
Unlike continuous Markovian stochastic volatility models, which are not able to fully describe the steep implied volatility skew of short-maturity options in equity markets, 
rough volatility models have shown accurate fit for this crucial feature. 
Within rough volatility, the rough Bergomi model~\cite{Bayer} is one of the simplest, 
yet decisive frameworks to harness the power of the roughness for pricing purposes.
We show how to adapt our functional quantization setup to this case.

\subsection{The generalized Bergomi model}\label{subsect_Bergomi}
We work here with a slightly generalised version of the rough Bergomi model, defined as
\begin{equation}\label{eq:gBergomi}
	\left\{
	\begin{array}{rcll}
		X_t & = & \displaystyle -\frac12 \int_0^t {{\mathcal{V}_s}} ds + \int_0^t \sqrt{{{\mathcal{V}_s}}} dB_s, & X_0=0,\\
		\mathcal{V}_t & = & \displaystyle v_0(t) \exp \left\{\gamma Z_t - \frac{\gamma^2}{2}\int_0^t K(t-s)^2 ds \right\}, & \mathcal{V}_0 >0,
	\end{array}
	\right.
\end{equation}
where $X$ is the log-stock price, $\mathcal{V}$ the instantaneous variance process
driven by the Gaussian Volterra process $Z$ in~\eqref{GV}, 
$\gamma>0$ and~$B$ is a Brownian motion defined as $B := \rho W + \sqrt{1-\rho^2} W^{\perp}$
for some correlation $\rho \in [-1,1]$ and $W, W^\perp$ orthogonal Brownian motions.
The filtered probability space is therefore taken as  $\cF_t=\cF_t^W \lor \cF_t^{W^\perp}$, $t \geq 0$.
This is a non-Markovian generalization of Bergomi's second generation stochastic volatility model~\cite{Bergomi}, letting the variance be driven by a Gaussian Volterra process instead of a standard Brownian motion.
Here, $v_T(t)$ denotes the forward variance for a remaining maturity~$t$, observed at time~$T$. 
In particular, $v_0$ is the initial forward variance curve, assumed to be $\cF_0$-measurable.
Indeed, given market prices of variance swaps $\sigma_T^2(t)$ at time~$T$ with remaining maturity~$t$, 
the forward variance curve can be recovered as 
$v_T(t)=\frac{d}{dt}\left(t\sigma_T^2(t)\right)$, for all $t \geq 0$,
and the process $\{ v_s(t-s)\}_{0 \leq s \leq t}$ is a martingale for all fixed $t>0$.

\begin{remark}\label{rem:RLKernelBergomi}
With $K(u) = u^{H-\frac{1}{2}}$, $\gamma = 2\nu C_H$, for $\nu>0$, and $C_H:= \sqrt{\frac{2H\Gamma(3/2-H)}{\Gamma(H+1/2)\Gamma(2-2H)}}$,  we recover the standard rough Bergomi model~\cite{Bayer}. 
\end{remark}

\subsection{VIX Futures in the generalized Bergomi}\label{subsect_VIX_F}
We consider the pricing of VIX Futures (\url{www.cboe.com/tradable_products/vix/}) in the rough Bergomi model.
They are highly liquid Futures on the Chicago Board Options Exchange Volatility Index,
introduced on March 26, 2004, to allow for trading in the underlying VIX.
Each VIX Future represents the expected implied volatility for the 30 days following the expiration date of the Futures contract itself.
The continuous version of the VIX at time~$T$ is determined by the continuous-time monitoring formula 
\begin{align}\label{Eq:VIX^2_dyn}
	\VIX_T^2 
	: & =  \mathbb{E}_{T}\left[\frac{1}{\Delta} \int_{T}^{T+\Delta} d\langle X_s, X_s\rangle\right]
	= \frac{1}{\Delta} \int_T^{T+\Delta} \mathbb{E}[\mathcal{V}_s | \cF_T]ds\\
&	=  \frac{1}{\Delta} \int_T^{T+\Delta} \mathbb{E}_{T}\left[v_0(s) 
e^{\gamma Z_s - \frac{\gamma^2}{2}\int_0^s K(s-u)^2 du }\right]ds\\
&	=  \frac{1}{\Delta} \int_T^{T+\Delta} v_0(s) e^{\gamma \int_0^T K(s-u)dW_u -\frac{\gamma^2}{2}\int_0^s K(s-u)^2 du } \mathbb{E}_{T}\left[e^{\gamma\int_T^s K(s-u) dW_u}\right]ds\\
&	=  \frac{1}{\Delta} \int_T^{T+\Delta} v_0(s) e^{\gamma \int_0^T K(s-u)dW_u -\frac{\gamma^2}{2}\int_0^s K(s-u)^2 du } e^{\frac{\gamma^2}{2}\int_T^s K(s-u)^2 du}ds,	
\end{align}
similarly to~\cite{Muguruza}, where~$\Delta$ is equal to~$30$ days,
and we write $\mathbb{E}_{T}[\cdot] := \mathbb{E}[\cdot |\cF_T]$
(dropping the subscript when $T=0$).
Thus, the price of a  VIX Future with maturity  $T$ is given by 
\begin{equation}\label{price_VIX}
		\mathcal{P}_T 
	    : = \mathbb{E}\left[\VIX_T\right]
		=\mathbb{E}\left[\left(  \frac{1}{\Delta}\int_T^{T + \Delta}v_0(t)e^{\gamma Z^{T,\Delta}_t+\frac{\gamma^2}{2}\left(\int_{0}^{t-T} K(s)^2 ds - \int_{0}^{t} K(s)^2 ds\right)} dt\right)^{\frac{1}{2}}\right],
\end{equation}
where the process $(Z^{T,\Delta}_t)_{t \in [T,T+\Delta]}$ is given by
\begin{equation}\label{eq:def_V_HTDelta}
	Z^{T,\Delta}_t = \int_0^{T}K(t-s)dW_s, \qquad t \in [T,T+\Delta].
\end{equation}

To develop a functional quantization setup for VIX Futures, 
we need to quantize the process~$Z^{T,\Delta}$, 
which is close, yet slightly different, from the  Gaussian Volterra process~$Z$ in~\eqref{GV}.

\subsection{Properties of $Z^{T}$}\label{Subsec_VH}
To retrieve the same setting as above, we normalize the time interval to $[0,1]$, that is $T+\Delta = 1$.
Then, for 
$T $ fixed, 
we define the process $Z^{T}:= Z^{T,1-T}$ as
\begin{equation}\label{V_HT}
	Z^{T}_t
	:= \int_0^T K(t-s)dW_s
	, \qquad t\in [T,1],
\end{equation}
which is well defined by the square integrability of~$K$.
By definition, the process $Z^{T}$ is centered Gaussian
and It\^o isometry gives its covariance function as
$$
R_{Z^{T}}(t,s) = \int_0^T K(t-u) K(s-u) du, \qquad t,s \in[T,1].
$$
Proceeding as previously, we introduce a Gaussian process with same mean and covariance as those of~$Z^{T}$, 
represented as a series expansion involving standard Gaussian random variables;
from which product functional quantization follows.
It is easy to see that the process~$Z^{T}$ has continuous trajectories.
Indeed, $ (Z_t^T-Z_s^T)^2 \leq \mathbb{E}[|Z_t-Z_s|^2| \mathcal F^W_T] $, by conditional Jensen's inequality since $Z^T_t = \mathbb E [Z_t | \mathcal F^W_T ]$. Then, applying tower property, for any $T \leq s < t \leq 1$,
\begin{eqnarray*}
	\mathbb{E}\left[\left|Z^{T}_t - Z^{T}_s\right|^2\right]
	\leq \mathbb{E}\left[|Z_t - Z_s|^2\right],
\end{eqnarray*}
and therefore the H-H\"older regularity of~$Z$ (Section~\ref{Sect_RL}) implies that of~$Z^T$.


\subsubsection{Series expansion}\label{VHT:series_repr}
Let $\{\xi_n\}_{n\geq1}$ be an i.i.d. sequence of standard Gaussian and $\{\psi_n\}_{n\geq1}$ the orthonormal basis of $L^2[0,1]$ from~\eqref{eq_psi_n}. 
Denote by  $\mathcal{K}^T(\cdot)$ the operator from $L^2[0,1]$ to $ \mathcal{C}[{T},1]$ that associates to each $f \in L^2[0,1]$,
\begin{equation}\label{eq:KT}
	\mathcal{K}^T[f](t)
	:=\int_0^T K(t-s)f(s)ds
	, \quad t \in [T,1].
\end{equation}
We define the process  $Y^{T}$ as (recall the analogous ~\eqref{series}):
\begin{equation}\label{series_T}
	Y^{T}_t := \sum_{n\geq 1} \mathcal{K}^T[\psi_n](t) \xi_n, \qquad t \in [T,1].
\end{equation}

The lemma below follows from the corresponding results in Remark~\ref{gau} and Lemma~\ref{lem:covY}:

\begin{lemma}\label{lem:covY_T}
	The  process $Y^{T}$ is centered, Gaussian and with covariance function 
	\[
		R_{Y^{T}}(s,t)
		:= \mathbb{E}\left[Y^{T}_s Y^{T}_t\right] 
		= \int_0^{T} K(t-u) K(s-u) du, \qquad \text{for all }s, t \in [T,1].
	\]
\end{lemma}

To complete the analysis of~$Z^{T}$, we require an analogue version of Assumption~\ref{AssumptionAB}.
\begin{assumption}\label{Assumption_V_H}
Assumption~\ref{AssumptionAB} holds for the sequence
$(\mathcal{K}^T[\psi_n])_{n\geq 1}$ on $[T,1]$
with the constants~$C_1$ and~$C_2$ depending on~$T$.
\end{assumption}


\subsection{The truncated RL case}\label{Sect_trunc_RL}
We again pay special attention to the RL case, 
for which the operator~\eqref{eq:KT} reads, for each $n \in \NN$,
$$
\mathcal{K}_H^T[\psi_n](t):= \int_0^T (t-s)^{H-\half} \psi_n(s) ds,
\qquad\text{for all }t \in [T,1],
$$
and satisfies the following, proved in Appendix~\ref{App_P_Lemma_V_H}:

\begin{lemma}\label{Lemma_V_H}
The functions $\{ \mathcal{K}_H^T[\psi_n] \}_{n \geq 1}$ satisfy Assumption~\ref{Assumption_V_H}.
\end{lemma}
A key role in this proof is played by an intermediate lemma, 
proved in Appendix~\ref{App_P_lem:Representation}, which provides a convenient representation for the integral $\int_{0}^{T}(t-u)^{H-\half}\E^{\I\pi u}du$, $t\geq T \geq 0$,
in terms of the generalised Hypergeometric function~${}_{1}{F}_{2}(\cdot)$.
\begin{lemma}\label{lem:Representation}
	For any  $t\geq T \geq 0$, the representation
$$
\int_{0}^{T}(t-u)^{H-\half}\E^{\I\pi u}du
		 = {\E^{\I\pi t}}\left[\Big(\zeta_{\half}(t, h_1) - \zeta_{\half}((t-T), h_1)\Big)
- \I\pi \Big(\zeta_{\frac{3}{2}}(t, h_2) - \zeta_{\frac{3}{2}}((t-T), h_2)\Big)\right]
$$
holds, where 
$h_1:= \half (H +\half)$ and $h_2 = \half+h_1$, $\chi(z) := -\frac{1}{4}\pi^2 z^2$ and
	\begin{equation}\label{def_zeta_khz}
		\zeta_{k}(z,h)
		:= \displaystyle \frac{z^{2h}}{2h}{}_{1}{F}_{2}\left(h; k,1+h; \chi(z)\right),
 \qquad\text{for } k\in\left\{\frac{1}{2}, \frac{3}{2}\right\}.
	\end{equation}
\end{lemma}

\begin{remark}\label{K_HT_repre_zeta}
The representation in Lemma~\ref{lem:Representation} can be exploited to obtain an explicit formula for $\mathcal{K}_H^T[\psi_n](t)$, $t \in [T,1]$ and $n \in \NN$:
	\begin{eqnarray*}
		& \mathcal{K}_H^T[\psi_n](t)
		=  \frac{\sqrt{2}}{m^{H+\half}}\int_0^{mT}(mt-u)^{H-\half}\cos(\pi u)du
		=  \frac{\sqrt{2}}{m^{H+\half}}\Re\left\{\int_0^{mT}(mt-u)^{H-\half}\E^{\I\pi u}du\right\}&\\
		& = \frac{\sqrt{2}}{m^{H+\half}}\Re\Bigg\{ {\E^{\I\pi mt}}\bigg[\Big(\zeta_{\half}(mt, h_1) - \zeta_{\half}(m(t-T), h_1)\Big)
- \I\pi \Big(\zeta_{\frac{3}{2}}(mt, h_2) - \zeta_{\frac{3}{2}}(m(t-T), h_2)\Big)\bigg] \Bigg\}& \\
		& = \frac{\sqrt{2}}{m^{H+\frac{1}{2}}}\bigg\{ \cos(mt\pi)\Big(\zeta_{\half}(mt, h_1) - \zeta_{\half}(m(t-T), h_1)\Big) +  \pi \sin(mt \pi) \Big(\zeta_{\frac{3}{2}}(mt, h_2) - \zeta_{\frac{3}{2}}(m(t-T), h_2)\Big) \bigg\},&
	\end{eqnarray*}
	with $m:=n-\half$ and $\zeta_{\half}(\cdot)$, $\zeta_{\frac{3}{2}}(\cdot)$  in~\eqref{def_zeta_khz}.
We shall exploit this in our numerical simulations. 
\end{remark}

\subsection{VIX Derivatives Pricing}

We can now introduce the quantization for the process~$Z^{T,\Delta}$, similarly to Definition~\ref{def_quant_Z_H}, recalling the definition of the set~$\cD_m^N$ in~\eqref{eq:SetD}:
\begin{definition}
	A \emph{product functional quantization for~$Z^{T,\Delta}$ of order $N$} is defined as
	\begin{equation}\label{p-quant_TDelta}
		\widehat{Z}^{T,\Delta,\dd}_t 	
		:= \sum_{n=1}^{m} \mathcal{K}^{T,\Delta}[\psi_n^{T,\Delta}](t) \widehat{\xi}_n^{d(n)}, \qquad t \in [T,T + \Delta],
	\end{equation}
	where $\dd \in\cD_m^N$, for some $m \in \NN$, and for every $ n \in \{1,\dots,m\},$ $\widehat{\xi}_n^{d(n)}$ is the (unique) optimal quadratic quantization of the Gaussian variable~$\xi_n$ of order $d(n)$.
\end{definition}
The sequence $\{\psi_n^{T,\Delta}\}_{n \in \NN}$ denotes the orthonormal basis of $L^2[0,T+\Delta]$ given by 
\begin{equation}\label{eq_psi_n_Delta}
	\psi_n^{T,\Delta}(t)=\sqrt{\frac{{2}}{T+\Delta}}\cos\left(\frac{t}{\sqrt{\lambda_n}(T+\Delta)}\right),
	\quad \text{ with }\lambda_n=\frac{4}{(2n-1)^2 \pi^2},
\end{equation}
and the operator $\mathcal{K}^{T,\Delta} : {L}^2[0,T+\Delta] \to  \mathcal{C}[T,T+\Delta]$ is defined for $f \in {L}^2[0,T+\Delta]$ as 
$$
	\mathcal{K}^{T,\Delta}[f](t)
	:=\int_0^T K(t-s)f(s)ds,
	\qquad t \in [T,T+\Delta].
$$
Adapting the proof of Proposition~\ref{Prop_staz_quant} it is possible to prove that these quantizers are stationary, too.
\begin{remark}
The dependence on $\Delta$ is due to the fact that the coefficients in the series expansion 
depend on the time interval $[T,T + \Delta]$.
\end{remark}
In the RL case for each $n \in \NN$, we can write, using Remark~\ref{K_HT_repre_zeta}, for any $t \in [T,T+\Delta]$:
\begin{align*}
	\mathcal{K}_H^{T,\Delta}[\psi_n^{T,\Delta}](t)&
	= \sqrt{\frac{{2}}{T+\Delta}}\int_0^T (t-s)^{H-\frac{1}{2}}\cos\left(\frac{s}{\sqrt{\lambda_n}(T+\Delta)}\right)ds,\\&
	= \frac{\sqrt{2}(T+\Delta)^H}{(n-1/2)^{H+\frac{1}{2}}}\int_0^{\frac{(n-1/2)}{T+\Delta}T}\Big(\frac{(n-1/2)}{T+\Delta}t-u\Big)^{H-\frac{1}{2}}\cos(\pi u)du\\
	&	=  \frac{\sqrt{2}(T+\Delta)^H}{(n-\half)^{H+\frac{1}{2}}}\bigg\{ \cos\Big(\frac{(n-\half)}{T+\Delta}t\pi\Big)\Big(\zeta_{\half}\Big(\frac{(n-\half)}{T+\Delta}t, h_1\Big) - \zeta_{\half}\Big(\frac{(n-\half)}{T+\Delta}(t-T), h_1\Big)\Big) \\&
\qquad+ \pi \sin\Big(\frac{(n-\half)}{T+\Delta}t \pi\Big) \Big(\zeta_{\frac{3}{2}}\Big(\frac{(n-\half)}{T+\Delta}t, h_2\Big) - \zeta_{\frac{3}{2}}\Big(\frac{(n-\half)}{T+\Delta}(t-T), h_2\Big)\Big) \bigg\}.
\end{align*}

We thus exploit $ \widehat Z^{T,\Delta, \dd}$ to obtain an estimation of $\VIX_T$ and of VIX Futures through the following
\begin{align}\label{quant_price_VIX}
		&\qquad \widehat{\VIX}_T^{\dd}  :=  \left(  \frac{1}{\Delta}\int_T^{T + \Delta}v_0(t)
\exp\left\{\gamma \widehat Z^{T,\Delta, \dd}_t + \frac{\gamma^2}{2}\left(\int_0^{t-T}K(s)^2ds-\int_0^tK(s)^2ds\right)\right\}dt\right)^{\half},\\
		&\qquad \widehat{ \mathcal{P}}_T^{\dd} 
		 :=  \mathbb{E}\left[\left(  \frac{1}{\Delta}\int_T^{T + \Delta}v_0(t)
\exp\left\{\gamma \widehat Z^{T,\Delta, \dd}_t + \frac{\gamma^2}{2}\left(\int_0^{t-T}K(s)^2ds-\int_0^tK(s)^2ds\right)\right\}dt\right)^{\half}\right].
\end{align}

\begin{remark}\label{rem_short_eq_for_price}
The expectation above reduces to the following deterministic summation, 
making its computation immediate:
	\begin{eqnarray*}
		\widehat{\mathcal{P}}_T^{\dd} 
		&=&\mathbb{E}\left[\left(  \frac{1}{\Delta}\int_T^{T + \Delta}v_0(t)
e^{\gamma \sum_{n=1}^{m} \mathcal{K}^{T,\Delta}[\psi_n^{T,\Delta}](t) \widehat{\xi}_n^{d(n)} +\frac{\gamma^2}{2}\left(\int_0^{t-T}K(s)^2ds-\int_0^tK(s)^2ds\right) } dt\right)^{\frac{1}{2}} \right]\\
		&=&\sum_{\underline{i} \in I^d} \left(  \frac{1}{\Delta}\int_T^{T + \Delta}v_0(t) e^{\gamma \sum_{n=1}^{m} \mathcal{K}^{T,\Delta}[\psi_n^{T,\Delta}](t) x_{i_n}^{d(n)}+ \frac{\gamma^2}{2}\left(\int_0^{t-T}K(s)^2ds-\int_0^tK(s)^2ds\right)} dt\right)^{\frac{1}{2}} \\
		& & \cdot \prod_{n=1}^m\mathbb{P}(\xi_n \in C_{i_n}(\Gamma^{d(n)})),
	\end{eqnarray*}
	where $\widehat{\xi}_n^{d(n)}$ is the (unique) optimal quadratic quantization of~$\xi_n$ of order $d(n)$, $C_j(\Gamma^{d(n)})$ is the $j$-th Voronoi cell relative to the $d(n)$-quantizer (Definition~\ref{def_quant_normal_rv}), with $j=1, \dots, d(n)$
 and $\underline{i}= (i_1,\dots,i_m) \in \prod_{j=1}^m \{1, \dots, d(j)\}$.
	In the numerical illustrations displayed in Section~\ref{Sect_Num_res}, we exploited Simpson rule to evaluate these integrals. In particular, we used  \emph{simps} function from \emph{scipy.integrate} with $300$ points.
\end{remark}

\subsection{Quantization error of VIX Derivatives}

The following $L^2$-error estimate is a consequence of Assumption~\ref{Assumption_V_H} \textbf{(B)} 
and its proof is omitted since it is analogous to that of  Proposition~\ref{prop:choice_m_d}:

\begin{proposition}\label{prop:choice_m_d_VHTD}
Under Assumption~\ref{Assumption_V_H}, for any $N \geq 1$, there exist $m_T^*(N) \in \NN$,
$C>0$ such that 
	\[
		\mathbb{E} \left[ \left\| \widehat{Z}^{T,\Delta,\dd^*_{T,N}}-Z^{T,\Delta} \right\|^2_{L^2[T,T+\Delta]} \right]^{\half} 
		\leq C \log(N)^{-H},
	\]
for $\dd^*_{T,N} \in \cD_{m_T^*(N)}^{N}$ and  with, for each $n=1, \dots, m_T^*(N)$,
$$
d^*_{T,N}(n)
	= \Big\lfloor  N^{\frac{1}{m_T^*(N)}} n^{-(H+\frac{1}{2})} \left(m_T^*(N)!\right)^{\frac{2H+1}{2m_T^*(N)}}  \Big\rfloor.
$$
Furthermore $m_T^*(N)= \mathcal{O}(\log(N))$.
\end{proposition}

As a consequence, we have the following error quantification for European options on the VIX:

\begin{theorem}\label{Thm_pricing_err}
Let $F: \R \to \R$ be a globally Lipschitz-continuous function and $\dd\in \NN^{m}$ for some $m \in \NN$. 
There exists $\mathfrak{K}>0$ such that
	\begin{equation}\label{ineq1_Thm_pricing_err}
		\left|\mathbb{E}\left[F\left(\VIX_T\right)\right] - \mathbb{E}\left[F\left(\widehat{\VIX}_T^{\dd}\right)\right]\right|
		\leq \mathfrak{K}\  \mathbb{E}\left[\left\| Z^{T,\Delta}- \widehat{Z}^{T,\Delta,\dd}\right\|_{L^2([T,T+\Delta])}^2\right]^{\half}.
	\end{equation}
	Furthermore, for any $N \geq 1$, there exist $m_T^*(N) \in \NN$ and $\mathfrak{C}>0$ such that, with $\dd^*_{T,N} \in \cD_{m_T^*(N)}^{N}$, 
	\begin{equation}\label{ineq2_Thm_pricing_err}
		\left|\mathbb{E}\left[F\left(\VIX_T\right)\right] - \mathbb{E}\left[F\left(\widehat{\VIX}_T^{\dd^*_{T,N}}\right)\right]\right|
		\leq \mathfrak{C} \log(N)^{-H}.
	\end{equation}
\end{theorem}

The upper bound in~\eqref{ineq2_Thm_pricing_err} is an immediate consequence of ~\eqref{ineq1_Thm_pricing_err} and Proposition~\ref{prop:choice_m_d_VHTD}. The proof of~\eqref{ineq1_Thm_pricing_err} is much more involved and is postponed to Appendix~\ref{App_Thm_pricing_err}.

\begin{remark}\ 
\begin{itemize}
\item When $F(x)=1$, we obtain the price of VIX Futures and the quantization error
$$
\left|\mathcal{P}_T - \widehat{\mathcal{P}}_T^{\dd}\right| 
		\leq \mathfrak{K} \  \mathbb{E}\left[\left\| Z^{T,\Delta}- \widehat{Z}^{T,\Delta,\dd}\right\|_{L^2([T,T+\Delta])}^2\right]^{\half},
$$
and, for any $N \geq 1$, Theorem~\ref{Thm_pricing_err} yields the existence of~$m_T^*(N) \in \NN$, $\mathfrak{C}>0$ such that	
$$
\left|\mathcal{P}_T - \widehat{\mathcal{P}}_T^{\dd^*_{T,N}}\right| 
\leq \mathfrak{C} \log(N)^{-H}.
$$
\item Since the functions $F(x):=(x-K)_{+}$ and $F(x):=(K-x)_{+}$ are globally Lipschitz continuous, the same bounds apply for European Call  and Put options on the VIX.
\end{itemize}
\end{remark}

\section{Numerical results for the RL case }\label{Sect_Num_res}
We now test the quality of the quantization on the pricing of VIX Futures in the standard rough Bergomi model, considering the RL kernel in Remark~\ref{rem:RLKernelBergomi}.

\subsection{Practical considerations for $m$ and $\bold{d}$}\label{sect_m_d_in_pract}

Proposition~\ref{prop:choice_m_d} provides, for any fixed $N \in \NN,$ some indications on~$m^*(N)$ and~$\dd^*_N\in \mathcal{D}_m^N$ (see~\eqref{eq:SetD}), for which the rate of convergence of the quantization error is $\log(N)^{-H}$.
We present now a numerical algorithm to compute the optimal parameters.
For a given number of trajectories $N \in \NN$, the problem is equivalent to finding $m \in \NN$ 
and $\dd\in \cD_m^N$ such that $\mathbb{E}[ \|Z^H-\widehat{Z}^{H,\dd}\|_{L^2[0,1]}^2]$ is minimal. 
Starting from~\eqref{eq_for_m_d_app} and adding and subtracting the quantity $\sum_{n=1}^m (\int_0^1 \mathcal{K}_H[\psi_n](t)^2 dt)$, we obtain
\begin{align}\label{esti}
\mathbb{E}\left[\left\|Z^H-\widehat{Z}^{H, \dd}\right\|_{L^2[0,1]}^2 \right]
& =\sum_{n=1}^m \left(\int_0^1 \mathcal{K}_H[\psi_n](t)^2 dt \right) [\err^{d(n)}( \xi_n )]^2 +\sum_{k \geq m+1}\int_0^1 \mathcal{K}_H[\psi_k](t)^2 dt\nonumber\\
& =\sum_{n=1}^m \left(\int_0^1 \mathcal{K}_H[\psi_n](t)^2 dt \right) \left\{\left[\err^{d(n)}( \xi_n )\right]^2 - 1\right\} +\sum_{k \geq 1}\int_0^1 \mathcal{K}_H[\psi_k](t)^2 dt,
ix\end{align}
where $\err^{d(n)}(\xi_n)$ denotes the optimal quadratic quantization error for the quadratic quantizer of order~$d(n)$ of the standard Gaussian random variable~$\xi_n$
(see Appendix~\ref{App_P_prop:choice_m_d} for more details).  
Notice that the last term on the right-hand side of~\eqref{esti} does not depend on~$m$, nor on~$\dd$.
We therefore simply look for~$m$ and $\dd$ that minimize
\[
A(m, \dd) := \sum_{n=1}^m \left(\int_0^1 \mathcal{K}_H[\psi_n]^2(t) dt \right) 
\left([\err^{d(n)}( \xi_n )]^2 - 1\right).
\]
This can be easily implemented: the functions $\mathcal{K}_H[\psi_n]$ can be obtained numerically from the Hypergeometric function
and the quadratic errors $\err^{d(n)}(\xi_n)$ are available at \url{www.quantize.maths-fi.com/gaussian_database}, for $d(n)\in\{1,\dots, 5999\}$.
The algorithm therefore reads as follows
\begin{enumerate}
	\item[(i)] fix $m$;
	\item[(ii)] minimize $A(m,\dd)$ over $\dd \in \mathcal{D}_m^N$ and call it~$\widetilde{A}(m)$;
	\item[(iii)] minimize $\widetilde A(m)$ over $m \in \NN$.
\end{enumerate}
The results of the algorithm for some reference values of $N \in \NN$ are available in  Table~\ref{tab:table11}, 
where $\overline{N}_{traj}:= \prod_{i=1}^{\overline{m}(N)}\overline{d}_N(i)$ represents the number of trajectories actually computed in the optimal case.
In Table~\ref{tab:table2}, we compute the rate optimal parameters derived in Proposition~\ref{prop:choice_m_d}:
the column `Relative error' contains the normalized difference between the $L^2$-quantization error made with the optimal choice of $\overline{m}(N)$ and $\overline{\dd}_N$ in Table~\ref{tab:table11} and the $L^2$-quantization error made with
$m^*(N)$ and ${\dd}^*_N$ of the corresponding line of the table, namely
$\frac{|\|Z^H-\widehat{Z}^{H, \overline{\dd}_N}\|_{L^2[0,1]}- \|Z^H-\widehat{Z}^{H, \dd^*_N}\|_{L^2[0,1]}|}{\|Z^H-\widehat{Z}^{H, \overline{\dd}_N}\|_{L^2[0,1]}}$.
In the column ${N}^*_{traj}:= \prod_{i=1}^{{m}^*(N)}{d}^*_N(i)$ we display the number of trajectories actually computed in the rate-optimal case.
The optimal quadratic vector quantization of a standard Gaussian of order~$1$ is the random variable identically equal to zero and so when $d(i)=1$ the corresponding term is uninfluential in the representation.

\begin{table}[h!]
  \begin{center}
  \caption{Optimal parameters.}
  \label{tab:table11}
    \begin{tabular}{c|c|c|c} 
      \textbf{$N$}  & $\overline{m}(N)$   & $\overline{\dd}_N$ & \textbf{$\overline{N}_{traj}$} \\
      \hline
      $10$  & $2$ & $5$ - $2$ & $10$ \\
      $10^2$  &  $4$  & $8$ - $3$  - $2$ - $2$ & $96$ \\
      $10^3$  &  $6$  & $10$ - $4$ - $3$ - $2$ - $2$ - $2$ & $960$\\
      $10^4$  &  $8$  & $10$ - $5$ - $4$ - $3$ - $2$ - $2$ - $2$ - $2$ & $9600$\\
      $10^5$ & $10$  & $14$ - $6$ - $4$ -  $3$ - $3$ - $2$ - $2$ - $2$ - $2$ - $2$ & $96768$ \\
      $10^6$  & $12$  & $14$ - $6$ - $5$ - $4$ - $3$ - $3$ - $2$ - $2$ - $2$ - $2$ - $2$ - $2$ & $967680$\\
    \end{tabular}
  \end{center}
\end{table}

\begin{table}[h!]
  \begin{center}
    \caption{Rate-optimal parameters.}
    \label{tab:table2}
    \begin{tabular}{c|c|c|c|c} 
      \textbf{$N$}  & $m^*(N)=\lfloor\log(N)\rfloor$ & \textbf{Relative error}  & $\dd^*_N$ & \textbf{$N_{traj}^*$} \\
      \hline
      $10$  & $2$ & $2.75\%$ & $3$ - $2$ & $6$ \\
      $10^2$ &  $4$ & $1.30\%$ & $5$ - $3$  - $2$ - $2$ & $60$ \\
      $10^3$  &  $6$ & $1.09\%$ & $6$ - $4$ - $3$ - $2$ - $2$ - $2$ & $576$\\
      $10^4$ &  $9$ & $3.08\%$ & $6$ - $4$ - $3$ - $2$ - $2$ - $2$ - $2$ - $1$ - $1$ & $1152$ \\
      $10^5$ & $11$ & $3.65\%$ & $7$ - $4$ - $3$ -  $3$ - $2$ - $2$ - $2$ - $2$ - $1$ - $1$ - $1$ & $4032$ \\
      $10^6$ & $13$ & $2.80\%$ & $8$ - $5$ - $4$ -  $3$ - $3$ - $2$ - $2$ - $2$ - $2$ - $2$ - $1$ - $1$ - $1$ & $46080$ \\
    \end{tabular}
  \end{center}
\end{table}

\begin{table}[h!]
  \begin{center}
    \label{tab:table1}
    \begin{tabular}{c|c|c|c|c} 
      \textbf{$N$} & $m^*(N)=\lfloor\log(N)\rfloor$ - 1 & \textbf{Relative error}  & $\dd^*_N$ & \textbf{$N_{traj}^*$} \\
      \hline
      $10$ & $1$ & $2.78\%$ & $10$ & $10$ \\
      $10^2$ &  $3$ & $1.13\%$ & $6$ - $4$  - $3$ & $72$ \\
      $10^3$ &  $5$ & $1.22\%$ & $7$ - $4$ - $3$ - $3$ - $2$  & $504$ \\
      $10^4$ &  $8$ & $1.35\%$ & $7$ - $4$ - $3$ - $3$ - $2$ - $2$ - $2$ - $2$ & $4032$ \\
      $10^5$ & $10$ & $2.29\%$ & $7$ - $5$ - $4$ -  $3$ - $2$ - $2$ - $2$ - $2$ - $2$ - $1$ & $13440$ \\
      $10^6$ & $12$ & $2.25\%$ & $8$ - $5$ - $4$ -  $3$ - $3$ - $2$ - $2$ - $2$ - $2$ - $2$ - $2$ - $1$ & $92160$  \\
    \end{tabular}
  \end{center}
\end{table}

\begin{table}[h!]
  \begin{center}
    \label{tab:table1}
    \begin{tabular}{c|c|c|c|c} 
      \textbf{$N$}  & $m^*(N)=\lfloor\log(N)\rfloor$ - 2 & \textbf{Relative error}  & $\dd^*_N$ & \textbf{$N_{traj}^*$}\\
      \hline
      $10^2$ &  $2$ & $2.53\%$ & $12$ - $8$ & $96$ \\
      $10^3$ &  $4$ & $1.44\%$ & $9$ - $5$ - $4$ - $3$  & $540$  \\
      $10^4$ &  $7$ & $1.46\%$ & $7$ - $5$ - $4$ - $3$ - $2$ - $2$ - $2$ & $3360$ \\
      $10^5$ & $9$ & $1.57 \%$ & $8$ - $5$ - $4$ -  $3$ - $3$ - $2$ - $2$ - $2$ - $2$  & $23040$\\
      $10^6$ & $11$ & $1.48 \%$ & $9$ - $6$ - $4$ -  $3$ - $3$ - $3$ - $2$ - $2$ - $2$ - $2$ - $2$ & $186624$ \\
    \end{tabular}
  \end{center}
\end{table}

\subsection{The functional quantizers}

The computations in Section~\ref{Sect_RL} and~\ref{section_quant} for the RL process, respectively the ones in Section~\ref{Subsec_VH} and~\ref{Sect_trunc_RL} for~$Z^{H,T}$, 
provide a way to obtain the functional quantizers of the processes. 

\subsubsection{Quantizers of the RL process}

For the RL process, Definition~\ref{def_quantizer_Z} shows that its quantizer 
is a weighted Cartesian product of grids of the one-dimensional standard Gaussian random variables. 
The time-dependent weights $\mathcal{K}_H[\psi_n](\cdot)$ are computed using~\eqref{eq_rep_KH},
and for a fixed number of trajectories~$N$, suitable $\overline{m}(N)$ and $\overline{\dd}_N\in\cD_{\overline{m}(N)}^{N}$ are chosen according to the algorithm in Section~\ref{sect_m_d_in_pract}. 
Not surprisingly, Figures~\ref{fig_quant_5} show that as the paths of the process get smoother ($H$ increases) the trajectories become less fluctuating and shrink around zero.
For $H=0.5$, where the RL process reduces to the standard Brownian motion, 
we recover the well-known quantizer from~\cite[Figures 7-8]{Pages}.
This is consistent as in that case 
$\mathcal{K}_H[\psi_n](t)= \sqrt{\lambda_n}\sqrt{2}\sin\left(\frac{t}{\sqrt{\lambda_n}}\right)$, and so~$Y^H$ is the Karhuenen-Lo\`eve expansion for the Brownian motion~\cite[Section 7.1]{Pages}.

\begin{figure}[htbp]
\begin{minipage}[b]{0.47\textwidth}
\centering
\includegraphics[scale=0.5]{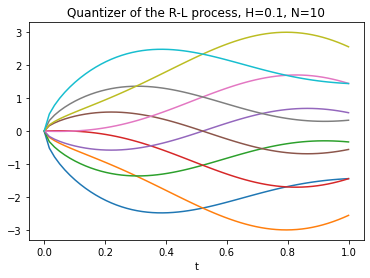}
\end{minipage}
\begin{minipage}[b]{0.47\textwidth}
\centering
\includegraphics[scale=0.5]{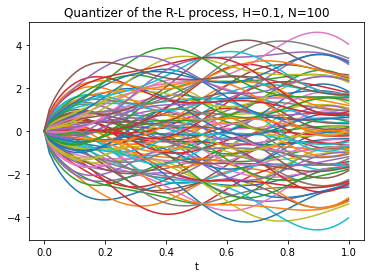}
\end{minipage}
\begin{minipage}[b]{0.47\textwidth}
\centering
\includegraphics[scale=0.5]{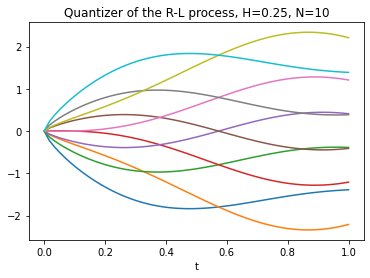}
\end{minipage}
\begin{minipage}[b]{0.47\textwidth}
\centering
\includegraphics[scale=0.5]{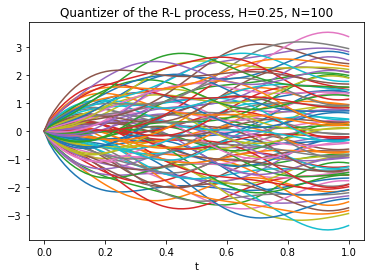}
\end{minipage}
\begin{minipage}[b]{0.47\textwidth}
\centering
\includegraphics[scale=0.5]{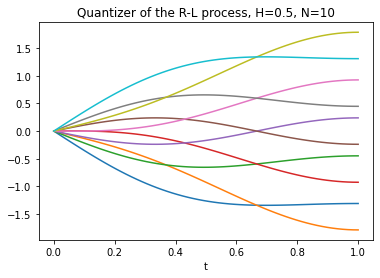}
\end{minipage}
\begin{minipage}[b]{0.47\textwidth}
\centering
\includegraphics[scale=0.5]{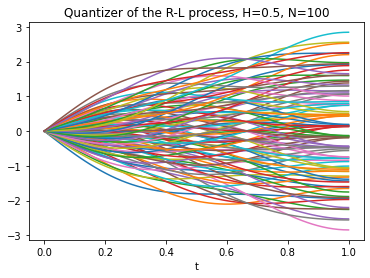}
\end{minipage}
\vspace{-0.3cm}
\caption{Product functional quantizations of the RL process with N-quantizers, 
for $H\in \{0.1, 0.25, 0.5\}$, for $N=10$ and $N=100$.}
\label{fig_quant_5}
\end{figure}

\subsubsection{Quantizers of $Z^{H,T}$}

A quantizer for $Z^{H,T}$ is defined analogously to that of $Z^H$ using Definition~\ref{def_quantizer_Z}. 
The weights $\mathcal{K}_H^T[\psi_n](\cdot)$ in the summation are available in closed form,
as shown in Remark~\ref{K_HT_repre_zeta}.
It is therefore possible to compute the $N$-product functional quantizer, for any $N \in \NN$, 
as Figure ~\ref{fig_quant_proc_V} displays.

\begin{figure}[htbp]
\begin{minipage}[b]{0.47\textwidth}
\centering
\includegraphics[scale=0.47]{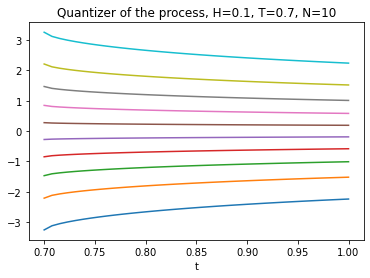}
\end{minipage}
\begin{minipage}[b]{0.47\textwidth}
\centering
\includegraphics[scale=0.47]{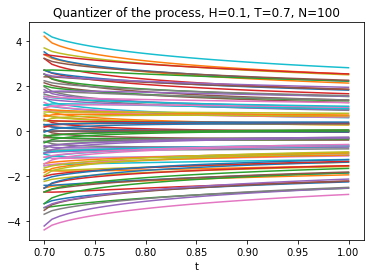}
\end{minipage}
\vspace{-0.3cm}
\caption{Product functional quantization of~$Z^{H,T}$ via N-quantizers, with $H=0.1$, $T=0.7$, for $N\in\{10, 100\}$.}
\label{fig_quant_proc_V}
\end{figure}

\newpage

\newpage

\subsection{Pricing and comparison with Monte Carlo}

In this section we show and comment some plots related to the estimation of prices of derivatives on the VIX and realized variance. 
We set the values $H=0.1$  and $\nu = 1.18778$ for the parameters and investigate three different initial forward variance curves~$v_0(\cdot)$, as in~\cite{Muguruza}:

\begin{itemize}
	\item[] Scenario 1. $v_0(t)= 0.234^2$;
	
	\item[] Scenario 2. $v_0(t)= 0.234^2(1 + t)^2$;
	
	\item[] Scenario 3. $v_0(t)= 0.234^2 \sqrt{1 + t}$.
\end{itemize}
The choice of such $\nu$ is a consequence of the choice $\eta=1.9$, consistently with~\cite{Bennedsen}, and of the relationship $\nu= \eta \frac{\sqrt{2 H}}{2 C_H}$.
In all these cases, $v_0$ is an increasing function of time, whose value at zero is close to the square of the reference value of $0.25$.

\subsubsection{VIX Futures Pricing}

One of the most recent and effective way to compute the price of VIX Futures is a Monte-Carlo-simulation method based on  Cholesky decomposition, for which we refer to~\cite[Section 3.3.2]{Muguruza}. It can be considered as a good approximation of the true price when the number~$M$ of computed paths is large. 
In fact, in~\cite{Muguruza} the authors tested three simulation-based methods (Hybrid scheme + forward Euler, Truncated Cholesky, SVD decomposition)  
and `all three methods seem to approximate the prices similarly well'. 
We thus consider the truncated Cholesky approach as a benchmark  
and take $M=10^6$ trajectories and $300$ equidistant point for the time grid. 

\begin{figure}[htbp]
\begin{minipage}[b]{0.47\textwidth}
\centering
\includegraphics[scale=0.5]{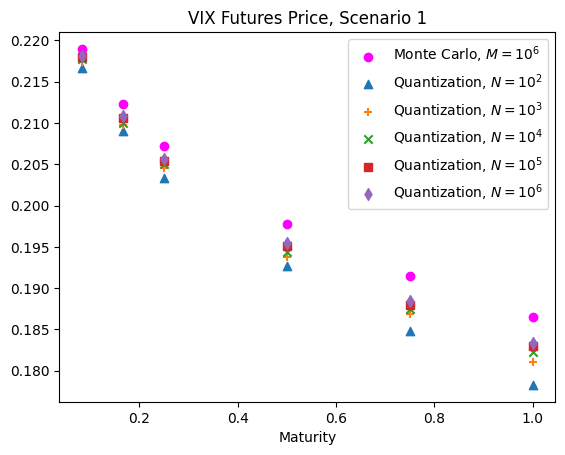}
\end{minipage}
\begin{minipage}[b]{0.47\textwidth}
\centering
\includegraphics[scale=0.5]{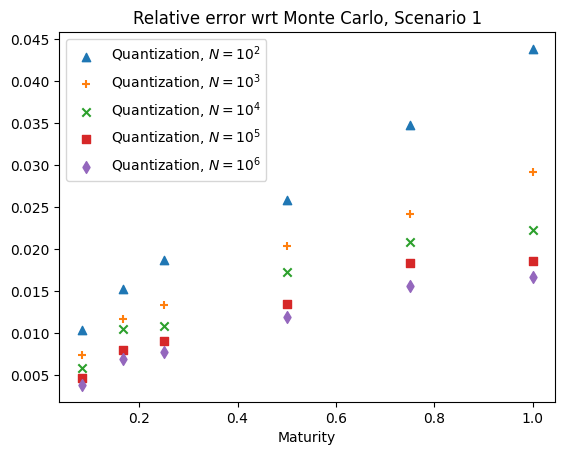}
\end{minipage}
\begin{minipage}[b]{0.47\textwidth}
\centering
\includegraphics[scale=0.5]{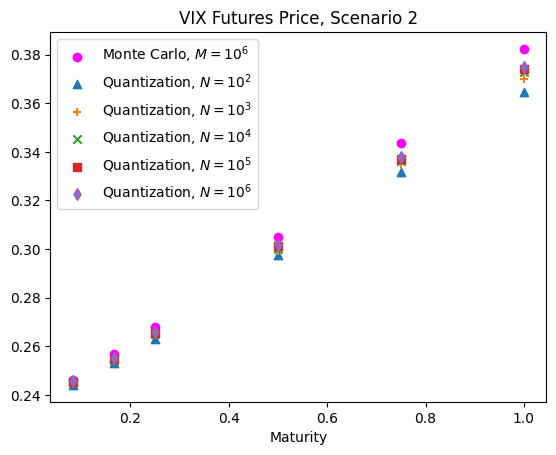}
\end{minipage}
\begin{minipage}[b]{0.47\textwidth}
\centering
\includegraphics[scale=0.5]{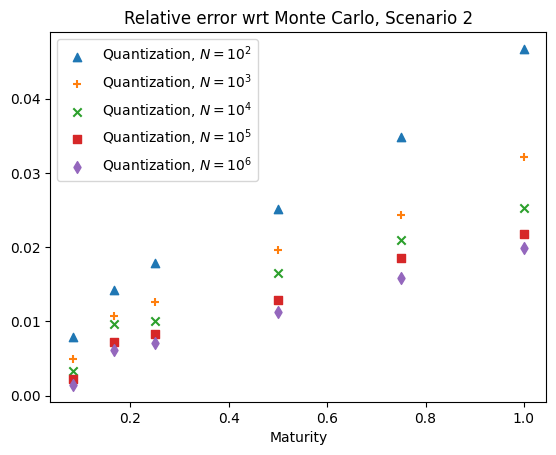}
\end{minipage}
\begin{minipage}[b]{0.47\textwidth}
\centering
\includegraphics[scale=0.5]{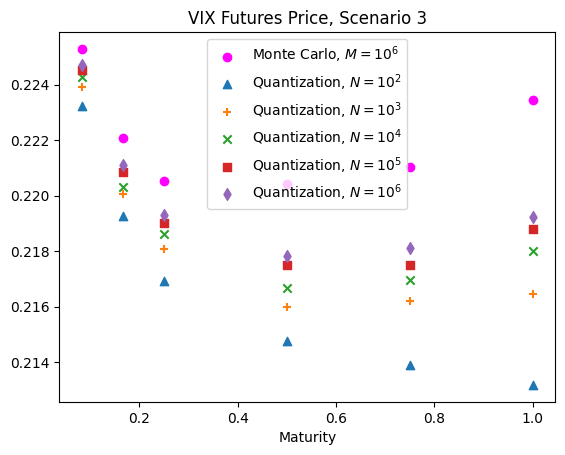}
\end{minipage}
\begin{minipage}[b]{0.47\textwidth}
\centering
\includegraphics[scale=0.5]{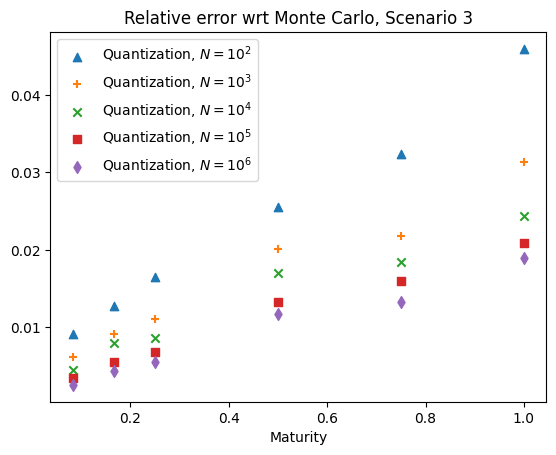}
\end{minipage}
\caption{VIX Futures prices (left) and relative error (right) computed with quantization and with Monte-Carlo as a function of the maturity~$T$, for different numbers of trajectories, for each forward variance curve scenario.}\label{fig_pr_sce3}
\end{figure}

In Figure~\ref{fig_pr_sce3},
we plot the VIX Futures prices as a function of the maturity~$T$, where $T$ ranges in $\{1, 2, 3, 6, 9, 12\}$ months
(consistently with actual quotations) on the left, and the corresponding relative error w.r.t. the Monte Carlo benchmark on the right.
It is clear that the quantization approximates the benchmark from below and that the accuracy increases with the number of trajectories.
\\ 
We highlight that the quantization scheme for VIX Futures can be sped up considerably by storing ahead the quantized trajectories for $Z^{H,T,\Delta}$, so that we only need to compute the integrations and summations in Remark~\ref{rem_short_eq_for_price}, which are extremely fast. 

\begin{table}[htbp!]
	\begin{center}
	\caption{Grid organization times (in seconds) as a function of the maturity (rows, in months) and of the number of trajectories (columns).}
	\label{table_grid_organization}
		\begin{tabular}{ |p{0.5cm}||p{2cm}|p{2cm}|p{2cm}|p{2cm}|p{2cm}|  }
		\hline
		\hline
		\multicolumn{6}{|c|}{\textbf{Grid organisation time}} \\
		\hline
		\hline
		  & $10^2$ & $10^3$ & $10^4$ & $10^5$ & $10^6$\\
		\hline
		\hline
		 1 &0.474 & 0.491 & 0.99 & 4.113 & 37.183 \\
		\hline
		 2 & 0.476 & 0.487 & 0.752 & 4.294 & 39.134 \\
		\hline 
		 3 & 0.617 & 0.536 & 0.826 & 4.197 & 37.744 \\
		\hline  
		 6 & 0.474 & 0.475 & 0.787 & 4.432 & 37.847 \\
		\hline  
		 9 & 0.459 & 0.6   & 0.858 & 3.73  & 41.988 \\
		\hline 
		 12 & 0.498 & 0.647 & 1.016 & 3.995 & 38.045 \\
		\hline
		\end{tabular}
		\end{center}
	\end{table}

Furthermore, the grid organization time itself is not that significant. 
In Table~\ref{table_grid_organization} we display the grid organization times (in seconds) as a function of the maturity (rows) expressed in months and of the number of trajectories (columns). From this table one might deduce that the time needed for the organization of the grids is suitable to be performed once per day (say every morning) as it should be for actual pricing purposes.
It is interesting to note that the estimations obtained with quantization (which is an exact method) are consistent in that they mimick the trend of benchmark prices over time even for very small values of~$N$. 
However, as a consequence of the variance in the estimations, 
the Monte Carlo prices are almost useless for small values of~$M$.
Moreover, improving the estimations with Monte Carlo requires to increase the number of points in the time grid with clear impact on computational time, while this is not the case with quantization since the trajectories in the quantizers are smooth.
Indeed, the trajectories in the quantizers are not only smooth but also almost constant over time, hence reducing the number of time steps to get the desired level of accuracy. 
Notice that here we may refer also to the issue of complexity related to discretization: a quadrature formula over $n$ points has a cost $\Oo(n)$, while the simulation with a Cholesky method over the same grid has cost $\Oo(n^2)$.
Finally, our quantization method does not require much RAM.
Indeed, all the simulations performed with quantization can be easily run on a personal laptop\footnote{The personal computer used to run the quantization codes has the following technical specifications: RAM: 8.00 GB, SSD memory: 512 GB, Processor: AMD Ryzen 7 4700U with Radeon Graphics 2.00 GHz.}, while this is not the case for the Monte Carlo scheme proposed here\footnote{The computer used to run the Monte Carlo codes is a virtual machine (OpenStack/Nova/KVM/Qemu, \url{www.openstack.org}) with the following technical specifications: RAM: 32.00 GB, CPU: 8 virtual cores, Hypervisor CPU: Intel(R) Xeon(R) CPU E5-2650 v3 @ 2.30GHz, RAM 128GB, Storage: CEPH cluster (\url{www.ceph.com}).}. 
For the sake of completeness, we also recall that combining Monte Carlo pricing of VIX futures/options with an efficient control variate speeds up the computations significantly~\cite{HJT2020}.

	\begin{figure}[htbp!]
	\begin{minipage}[b]{0.47\textwidth}
	\centering
	\includegraphics[scale=0.5]{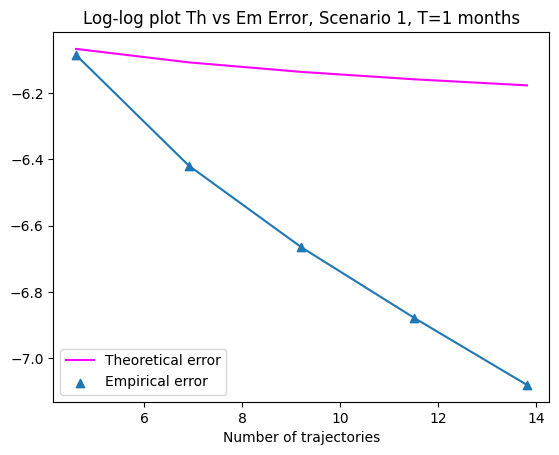}
	\end{minipage}
	\begin{minipage}[b]{0.47\textwidth}
	\centering
	\includegraphics[scale=0.5]{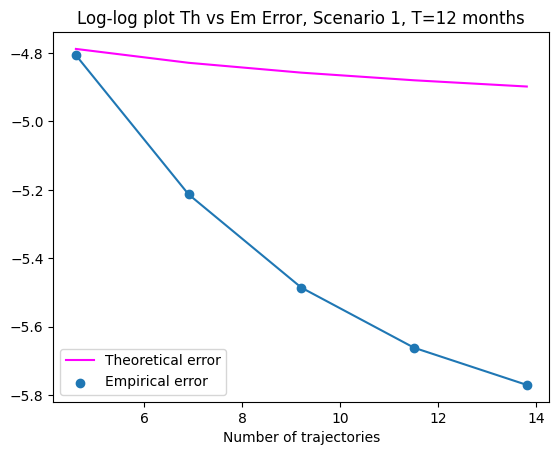}
	\end{minipage}
	\caption{Log-log (natural logarithm) plot of the empirical absolute error with the theoretically predicted one for Scenario 1, with $T \in \{ 1, 12\}$ months.}
	\label{fig_log_log_plot_th_vs_em_err}
	\end{figure}
	
	In Figure~\ref{fig_log_log_plot_th_vs_em_err}, we show some plots comparing the behaviour of the empirical error with the theoretically predicted one. 
	We have decided to display only a couple of maturities for the first scenario since the other plots are very similar.
	The figures display in a clear way that the order of convergence of the empirical error should be bigger than the theoretically predicted one: in particular, we expect it to be $\Oo(\log(N)^{-1})$.

\subsubsection{VIX Options Pricing}

	To complete the discussion on VIX Options pricing, we present in Figure~\ref{fig_call_VIX} the approximation of the prices of ATM Call Options on the VIX obtained via quantization as a function of the maturity $T$ and for different numbers of trajectories against the same price computed via Monte Carlo simulations with $M=10^6$ trajectories and $300$ equidistant point for the time grid, as a benchmarch.
	Each plot represents a different scenario for the initial forward variance curve. 
	For all scenarios, as the number~$N$ of trajectories goes to infinity, the prices  in Figure~\ref{fig_call_VIX} are clearly converging, and the limiting curve is increasing in the maturity, as it should be.

	\begin{figure}[htbp]
	\begin{minipage}[b]{0.47\textwidth}
	\centering
all	\includegraphics[scale=0.5]{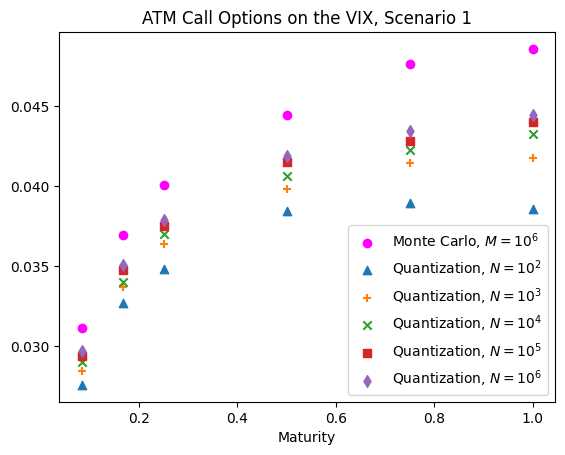}
	\end{minipage}
	\begin{minipage}[b]{0.47\textwidth}
	\centering
	\includegraphics[scale=0.5]{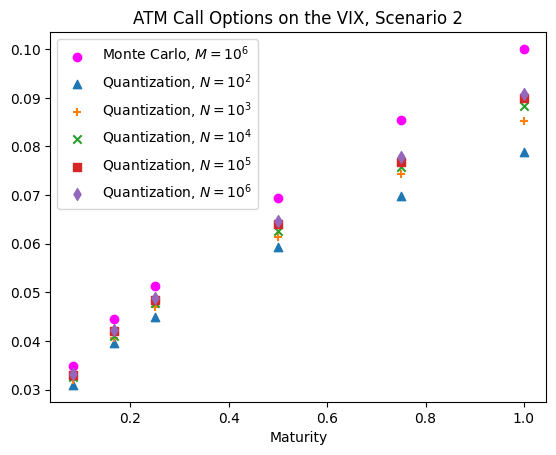}
	\end{minipage}
	\begin{minipage}[b]{0.47\textwidth}
	\centering
	\includegraphics[scale=0.5]{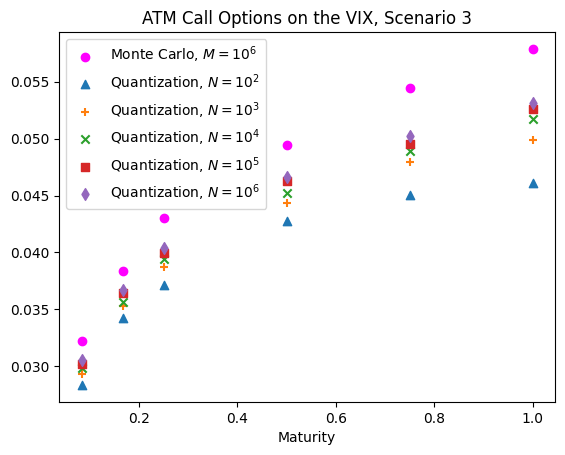}
	\end{minipage}
	\caption{Prices of ATM Call Options on the VIX via quantization.}
	\label{fig_call_VIX}
	\end{figure}

\subsubsection{Pricing of Continuously Monitored Options on Realized Variance}

    Product functional quantization of the process $(Z_t^H)_{t \in [0,T]}$ can be exploited for (meaningful) pricing purposes, too.
    We first price variance swaps, whose price is given by the following expression
     \begin{displaymath}
     	\mathfrak{S}_T :=\mathbb{E}\left[ \frac{1}{T}\int_0^T \mathcal{V}_t \D t \bigg| \mathcal{F}_0\right].
     \end{displaymath} 
     
     Let us recall that, in the rough Bergomi model,
     \begin{displaymath}
     \mathcal{V}_t = v_0(t)\exp \Big(2\nu C_H Z_t^H-\frac{\nu^2C_H^2}{H}t^{2H}\Big), 
     \end{displaymath}
     where $C_H= \sqrt{\frac{2H\Gamma(3/2-H)}{\Gamma(H+1/2)\Gamma(2-2H)}}$, $ \nu>0$ is an endogenous  constant and $v_0(t)$ being the initial forward variance curve.
	Thus, exploiting the fact that, for any fixed $t \in [0,T]$, $Z^H_t$ is distributed according to a centred Gaussian random variable with variance $\int_0^t (t-s)^{2H-1} ds = \frac{t^{2H}}{2H}$, the quantity $\mathfrak{S}_T$ can be explicitly computed: 
	\begin{displaymath}
	\mathfrak{S}_T = \frac{1}{T}\int_0^T v_0(t) \D t.
	\end{displaymath}
	This is particularly handy and provides us a simple benchmark.
	The price $\mathfrak{S}_T$ is, then, approximated via quantization through
	\begin{align}
	\widehat{\mathfrak{S}}_T^{\bold{d}}= \sum_{\underline{i} \in I^d} \left(\frac{1}{T}\int_0^T v_0(t) \exp \left(2\nu C_H \sum_{n=1}^m \mathcal{K}_H[\psi_n](t)x_{i_n}^{d(n)} - \frac{\nu^2 C_H^2}{H}t^{2H}\right) dt \right) \prod_{n=1}^m \mathbb{P}(\xi_n \in C_{i_n}(\Gamma^{d(n)})).
	\end{align}
	
	Numerical results are presented in Figure~\ref{fig_price_swaps_real_var}.
	On the left-hand side we display a table with the approximations (depending on $N$, the number of trajectories) of the price of a swap on the realized variance in Scenario 1, for $T=1$, and the true value $v_0=0.234^2$. On the right-hand side a log-log (natural logarithm) plot of the error against the function $c \log(N)^{-H}$, with $c$ being a suitable positive constant.
	For variance swaps the error is not performing very well. 
	It is indeed very close to the upper bound $c\log(N)^{-H}$ that we have computed theoretically. 
	One possible theoretical motivation for this behaviour lies in the difference between \emph{strong} and \emph{weak} error rates.
	Weak error and strong error do not necessarily share the same order of convergence, being the weak error faster in general.
	See~\cite{Bayer_rates_1,Bayer_rates_2, Gassiat} for recent developments on the topic in the rough volatility framework.
	For pricing purposes, we are interested in weak error rates. 	
	Indeed, the pricing error should in principle have the following form $\mathbb{E}[f ( Z^H ) ]- \mathbb{E}[f( \widehat{Z}^H )]$, where $\widehat{Z}^H $ is the process that we are using to approximate the original  $Z^H $ and $f$ is a functional that comes from the payoff function and that we can interpret as a test function. Thus, the functional $f$ has a smoothing effect. 
	On the other hand, the upper bound for the quantization error we have computed is a strong error rate.
	This theoretical discrepancy motivates the findings in Figure~\ref{fig_log_log_plot_th_vs_em_err} when pricing VIX Futures and other options on the VIX:  the empirical error seems to converge with order $\mathcal{O}(\log(N)^{-1})$, while the predicted order is $\mathcal{O}(\log(N)^{-H})$.
	The different empirical rates that are seen in  Figure~\ref{fig_log_log_plot_th_vs_em_err}  for VIX futures (roughly $\mathcal O(\log(N)^{-1})$)) and in Figure~\ref{fig_price_swaps_real_var} for variance swaps (much closer to $\mathcal O(\log(N)^{-H})$) could be also related to the different degree of pathwise regularity of the processes $Z$ and $Z^T$ . While $t \rightarrow Z_t = \int_0^t K(t - s) dWs$ is a.s. $(H - \epsilon)$-H\"{o}lder, for fixed $T$, the trajectories $t \rightarrow Z_t^T =
\int_0^T K(t - s) dWs$ of $Z^T$ are much smoother when $t \in (T, T + \Delta) $ and $t$ is bounded away from $T$.
When pricing VIX derivatives, we are quantizing almost everywhere a smooth Gaussian process (hence error rate of order $\log(N)^{-1})$, while when pricing derivatives on realized variance, we are applying quantization to a rough Gaussian process (hence error rate of order $\mathcal O(\log(N)^{-H})$), resulting in a deteriorated accuracy for
the prices of realized volatility derivatives such as the variance swaps in Figure~\ref{fig_price_swaps_real_var}.

	Furthermore, it can be easily shown that, for any $\bold{d} \in \mathcal{D}_m^N$ and for any $m, N \in \mathbb{N}$, with $m < N$,  $\widehat{\mathfrak{S}}_T^{\bold{d}}$ always provides a lower bound for the true price ${\mathfrak{S}}_T$.
	Indeed, since the quantizers $\widehat{Z}^{H,\bold{d}}$ of the process $Z^H$ are stationary (cfr. Proposition~\ref{Prop_staz_quant}), an application  of Remark~\ref{rem_staz_quant} to the convex function $f(x)= \exp(2\nu C_H x )$ together with the positivity of $v_0(t)\exp(-\frac{\nu^2C_H^2 t^{2H}}{H})$, for any $t \in [0,T]$, yields
	\begin{align*}
		\widehat{\mathfrak{S}}_T^{\bold{d}} 
		& = \mathbb{E}\left[ \frac{1}{T}\int_0^T v_0(t)\exp\left(-\frac{\nu^2C_H^2 t^{2H}}{H}\right)\exp\left(2\nu C_H \widehat{Z}_T^{H, \bold{d}}\right) \D t \bigg| \mathcal{F}_0\right]\\
		& = \frac{1}{T}\int_0^T v_0(t)\exp\left(-\frac{\nu^2C_H^2 t^{2H}}{H}\right)\mathbb{E}_0\left[\exp\left(2\nu C_H \widehat{Z}_T^{H, \bold{d}}\right)\right] \D t\\
		& \leq \frac{1}{T}\int_0^T v_0(t)\exp\left(-\frac{\nu^2C_H^2 t^{2H}}{H}\right)\mathbb{E}_0\left[\exp\left(2\nu C_H {Z}_T^{H}\right)\right] \D t
		= \mathfrak{S}_T.
	\end{align*}


\begin{figure}[htbp]
	
	\begin{minipage}[t]{0.47\textwidth}
	\vspace{-4 cm}

	\begin{tabular}{ |p{4cm}|p{1.2cm}|  }
		\hline
		 True price &  $0.0548$\\
		\hline
		\hline
		 Quantization,  $N=10^2$ & $0.0230$\\
		\hline
		 Quantization,  $N=10^3$ & $0.0246$\\
		\hline 
		 Quantization,  $N=10^4$ & $0.0257$\\
		\hline  
		 Quantization,  $N=10^5$ & $0.0266$\\
		\hline  
		 Quantization,  $N=10^6$ & $0.0273$\\
		\hline 
	\end{tabular}
	
	\end{minipage}
	\begin{minipage}[b]{0.47\textwidth}
	\centering
	\includegraphics[scale=0.5]{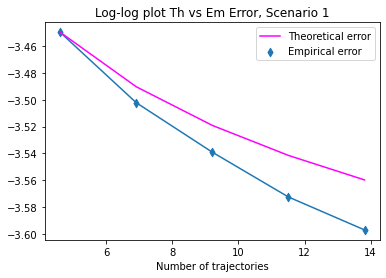}
	\end{minipage}
	\caption{Prices and errors for variance swaps.}
	\label{fig_price_swaps_real_var}
\end{figure}
	
	To complete this section, we plot  in Figure~\ref{fig_price_call_cm_real_var} approximated prices of European Call Options on the realized variance via quantization with $N \in \{10^2, 10^3,10^4,10^5, 10^6\}$ trajectories and via Monte Carlo with $M=10^6$ trajectories, as a benchmark. 
	In order to take advantage of the trajectories obtained, we compute the price of a realized variance Call option with strike $K$ and maturity $T=1$ as
$$    
\mathcal{C}(K,T)=\mathbb E\left[ \left(\frac{1}{T}\int_0^T \mathcal{V}_t \D t - K\right)_{+}\bigg| \mathcal F_0\right],
$$
and we approximate it via quantization through
$$
 \widehat{\mathcal{C}}^{\bold{d}}(K,T)= \sum_{\underline{i} \in I^d} \Bigg(\frac{1}{T}\int_0^T v_0(t) \exp \left(2\nu C_H \sum_{n=1}^m \mathcal{K}_H[\psi_n](t)x_{i_n}^{d(n)} - \frac{\nu^2 C_H^2}{H}t^{2H}\right) dt - K \Bigg)_{+} \prod_{n=1}^m \mathbb{P}(\xi_n \in C_{i_n}(\Gamma^{d(n)})).
$$	
	
	The three plots in Figure~\ref{fig_price_call_cm_real_var} display the behaviour of the price of a European Call on the realized variance as a function of the strike price $K$ (close to the ATM value) for the three scenarios considered before.
	
	\begin{figure}[htbp]
	\begin{minipage}[b]{0.47\textwidth}
	\centering
	\includegraphics[scale=0.5]{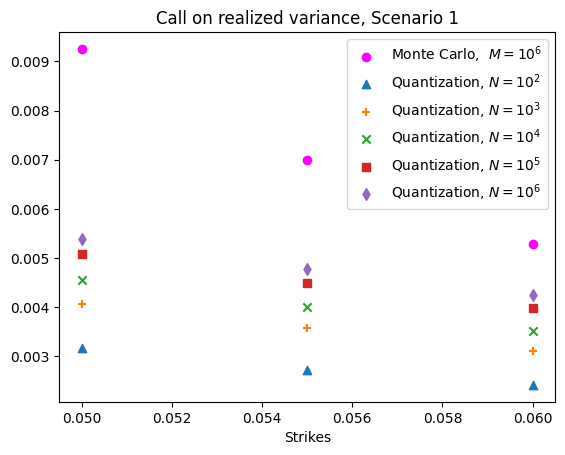}
	\end{minipage}
	\begin{minipage}[b]{0.47\textwidth}
	\centering
	\includegraphics[scale=0.5]{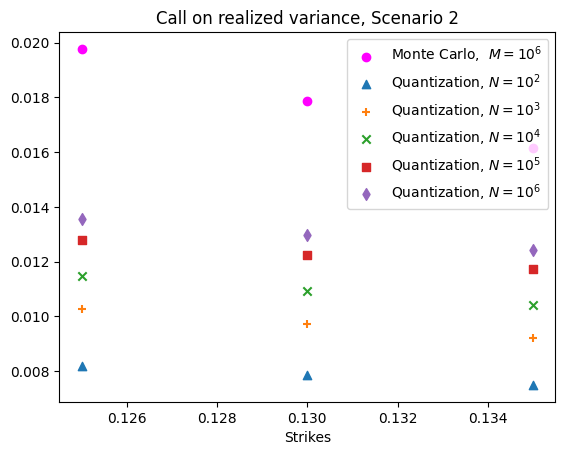}
	\end{minipage}
	\begin{minipage}[b]{0.47\textwidth}
	\centering
	\includegraphics[scale=0.5]{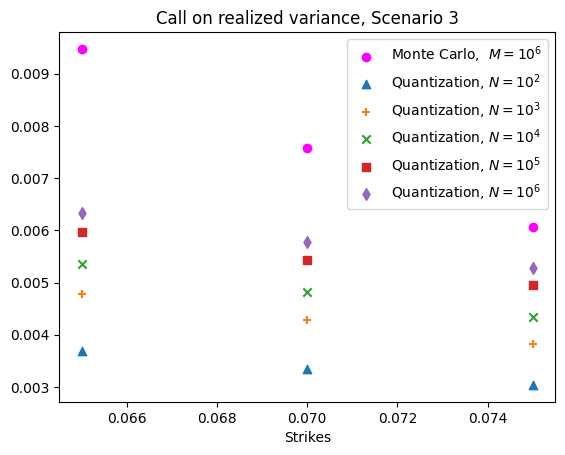}
	\end{minipage}
	\caption{Prices of European Call Option on realized variance computed via Monte Carlo with $M=10^6$ trajectories and via quantization with $N\in \{10^2, 10^3, 10^4, 10^5, 10^6\}$ trajectories, as a function of $K$.}
	\label{fig_price_call_cm_real_var}
	\end{figure}

\subsubsection{Quantization and MC comparison}

	In order to make a fair comparison between quantization and Monte Carlo simulations, we present a figure to display, for each methodology, the computational work needed for a given error tolerance for the pricing of VIX Futures.
	The plots in Figure~\ref{fig_true_err_comp} should be read as follows. 
	First, for any $M, N \in \{10^2, 10^3,10^4,10^5, 10^6\}$, we have computed the corresponding pricing errors: 
	$
	\varepsilon^{MC}(M):= | \textrm{Price}^{MC}(M) - \textrm{RefPrice}|$ and $  \varepsilon^{Q}(N) :=  | \textrm{Price}^{Q}(N) - \textrm{RefPrice}|
	$
	 where $\textrm{Price}^{MC}(M)$ is the Monte Carlo price obtained via truncated Cholesky with $M$ trajectories, $\textrm{Price}^{Q}(N) $ is the price computed via quantization with $N$ trajectories and $\textrm{RefPrice}$ comes from the lowerbound in Equation $(3.4)$ in \cite{Muguruza} and the associated computational time in seconds $t^{MC}(M)$ and $t^{Q}(N)$, respectively for Monte Carlo simulation and quantization.
Then, each point in the plot is associated either to a value of~$M$ in case of Monte Carlo (the circles in Figure~\ref{fig_true_err_comp}), or $N$ in case of quantization  (the triangles in Figure~\ref{fig_true_err_comp}), and its $x$-coordinate provides the absolute value of the associated pricing error, while its $y$-coordinate represents  the associated computational cost in seconds.
	
These plots lead to the following observations:
	\begin{itemize}
	\item For quantization, which is an exact method, the error is strictly monotone in the number of trajectories.
	
	\item When a small number of trajectories is considered, quantization provides a lower error with respect to Monte Carlo, at a comparable cost.
	
	\item For large numbers of trajectories Monte Carlo overcomes quantization both in terms of accuracy and of computational time.
	
	\end{itemize}
	To conclude, quantization can always be run with an arbitrary number of trajectories and furthermore for $N \in \{10^2, 10^3,10^4\}$ it leads to a lower error with respect to Monte Carlo, at a comparable computational cost, as it is visible from Figure \ref{fig_true_err_comp}. 
	This makes quantization particularly suitable to be used when dealing with standard machines, i.e., laptops with a RAM memory smaller or equal to $16$GB.

	\begin{figure}[htbp]
	\begin{minipage}[b]{0.47\textwidth}
	\centering
	\includegraphics[scale=0.5]{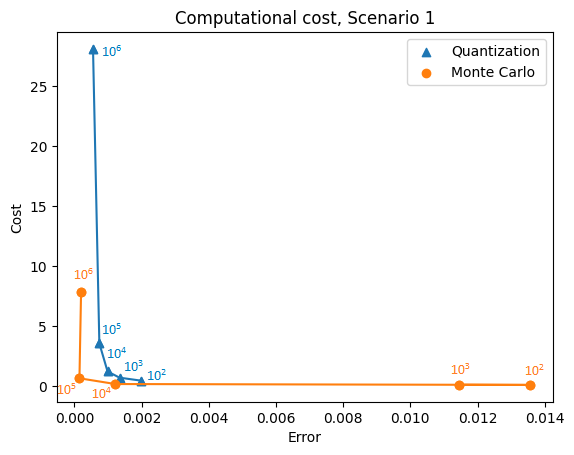}
	\end{minipage}
	\begin{minipage}[b]{0.47\textwidth}
	\centering
	\includegraphics[scale=0.5]{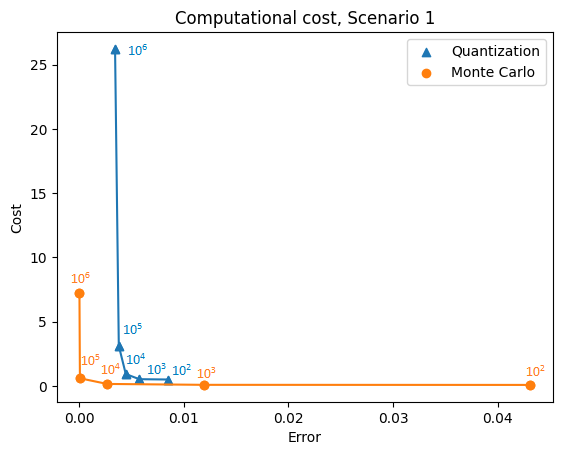}
	\end{minipage}
	\caption{Computational costs for quantization vs Monte Carlo for Scenario $1$, with $T=1$ month (left-hand side) and $T=12$ months (right-hand side). The number of trajectories, $M$ for Monte Carlo and $N$ for quantization, corresponding to a specific dot is displayed above it.}
	\label{fig_true_err_comp}
	\end{figure}


\section{Conclusion}
In this paper we provide, on the theoretical side, a precise and detailed result  on the convergence of product functional quantizers  of \emph{Gaussian Volterra processes}, showing that the $L^2$-error is of order $\log(N)^{-H}$, with~$N$ the number of trajectories and~$H$ the regularity index.

Furthermore, we explicitly characterize the rate optimal parameters, $m^*_N$ and $\bold{d}^*_N$, and we  compare them with the corresponding optimal parameters, $\overline{m}_N$ and $\overline{\bold{d}}_N$, computed numerically.

In the rough Bergomi model, we apply product functional quantization to the pricing of VIX options, with precise rates of convergence, and of options on realized variance, comparing those -- whenever possible -- to standard Monte Carlo methods.

The thorough numerical analysis carried out in the paper shows that unfortunately, despite the conceptual promise of functional quantization, while the results on the VIX are very promising, other types of path-dependent options seem to require machine resources way beyond the current requirements of standard Monte Carlo schemes, as shown precisely in the case of variance swaps.
While product functional quantization is an exact method, the analysis provided here does not however promise a bright future in the context of rough volatility. 
		It may nevertheless be of practical interest when machine resources are limited and indeed the results for VIX Futures pricing are strongly encouraging in this respect.
		Functional quantization for rough volatility can however be salvaged when used as a control variate tool to reduce the variance in classical Monte Carlo simulations.

\newpage

\appendix

\section{Proofs}

\subsection{Proof of Proposition~\ref{prop:choice_m_d}}\label{App_P_prop:choice_m_d}

Consider a fixed $N \geq 1$ and $(m,\dd)$ for $\dd\in\cD_m^N$. We have
\begin{align}\label{eq_for_m_d_app}
\mathbb{E}\left[\left\|Z - \widehat{Z}^{\dd} \right\|_{L^2[0,1]}^2\right]
&  =\mathbb{E}\left[\left\|\sum_{n\geq 1} \mathcal{K}[\psi_n](\cdot) \xi_n- \sum_{n=1}^{m} \mathcal{K}[\psi_n](\cdot)\widehat{\xi}_n^{d(n)} \right\|_{L^2[0,1]}^2 \right]\nonumber\\
& =\mathbb{E}\left[\left\|\sum_{n= 1}^m \mathcal{K}[\psi_n](\cdot) (\xi_n-\widehat{\xi}_n^{d(n)})+\sum_{k \geq m+1} \mathcal{K}[\psi_k](\cdot)\xi_k \right\|_{L^2[0,1]}^2 \right]\nonumber\\
& =\mathbb{E}\left[ \int_0^1\left| \sum_{n= 1}^m \mathcal{K}[\psi_n](t) (\xi_n-\widehat{\xi}_n^{d(n)})+\sum_{k \geq m+1} \mathcal{K}[\psi_k](t)\xi_k \right|^2 dt\right]\nonumber\\
& = \int_0^1 \left(\sum_{n= 1}^m \mathcal{K}[\psi_n]^2(t) \mathbb{E}\left[|\xi_n-\widehat{\xi}_n^{d(n)}|^2\right] +\sum_{k \geq m+1} \mathcal{K}[\psi_k]^2(t)\right) dt\nonumber \\
& = \int_0^1 \left( \sum_{n= 1}^m \mathcal{K}[\psi_n]^2(t) \err^{d(n)}( \xi_n )^2
			+\sum_{k \geq m+1} \mathcal{K}[\psi_k]^2(t) \right)  dt,
\end{align}
using Fubini's Theorem and the fact that $\{\xi_n\}_{n\geq 1}$ is a sequence of i.i.d. Gaussian and where
$\err^{d(n)}(\xi_n):= \inf_{(\alpha_1, \dots, \alpha_{d(n)}) \in \R^{d(n)}} \sqrt{\mathbb{E}[\min_{1 \leq i \leq d(n)} | \xi_n - \alpha_i|^2]}$.  
The Extended Pierce Lemma~\cite[Theorem 1(b)]{Pages}
ensures that $\err^{d(n)}( \xi_n  )\leq  \frac{L}{d(n)}$ for a suitable positive constant~$L$.
Exploiting this error bound and the property \textbf{(B)} for $\mathcal{K}[\psi_n]$ in Assumption~\ref{AssumptionAB}, we obtain
\begin{align}\label{est}
\mathbb{E}\left[ \|Z-\widehat{Z}^{\dd}\|_{L^2[0,1]}^2 \right]
& =\sum_{n=1}^m \left(\int_0^1 \mathcal{K}[\psi_n]^2(t) dt \right) \err^{d(n)}( \xi_n )^2
 + \sum_{k \geq m+1}\int_0^1 \mathcal{K}[\psi_k]^2(t) dt\\
& \leq C_2^2\left\{\sum_{n=1}^m n^{-(2H+1)} \err^{d(n)}( \xi_n )^2 +\sum_{k \geq m+1} k^{-(2H+1)}\right\}\\
& \leq C_2^2 \left\{\sum_{n=1}^m n^{-(2H+1)} \frac{L^2}{d(n)^2} + \sum_{k \geq m+1} k^{-(2H+1)}\right\}\\
& \leq \widetilde{C} \left( \sum_{n=1}^m  \frac{1}{n^{2H+1}d(n)^2} +  \sum_{k \geq m+1} k^{-(2H+1)}\right),
\end{align}
with $\widetilde{C} = \max \{L^2 C_2^2 , C_2^2  \}$.
	Inspired by ~\cite[Section 4.1]{Pages29}, we now look for an ``optimal'' choice of~$m \in \NN$
and $\dd\in \cD_m^N$. 
This reduces the error in approximating~$Z$ with a product quantization of the form in~\eqref{p-quant}.
Define the optimal product functional quantization ~$\widehat{Z}^{N,\star}$ of order~$N$
as the $\widehat{Z}^{\dd}$ which realizes the minimal error:
$$
\mathbb{E}\left[\left\|Z-\widehat{Z}^{N,\star}\right\|_{L^2[0,1]}^2 \right]
=\min \left\{\mathbb{E}\left[\left\|Z-\widehat{Z}^{\dd}\right\|_{L^2[0,1]}^2 \right], m \in \NN, \dd\in\cD_m^N\right\}.
$$
From~\eqref{est} we deduce
\begin{equation}\label{firstes}
\mathbb{E}\left[\left\|Z-\widehat{Z}^{N,\star}\right\|_{L^2[0,1]}^2 \right]
\leq \widetilde C \inf_{m \in \NN} \left\{ \sum_{k\geq m+1} \frac{1}{k^{2H+1}}
 + \inf \left\{\sum_{n=1}^m  \frac{1}{n^{2H+1}d(n)^2}, \dd\in\cD_m^N\right\} \right\}.
\end{equation}
	For any fixed $m \in \NN$ we associate to the internal minimization problem the one we get by relaxing the hypothesis that
	$d(n ) \in \NN$:
	\begin{equation}
\mathfrak{I} := \inf \Big\{  \sum_{n=1}^m  \frac{1}{n^{2H+1}z(n)^2}, \{z(n)\}_{n=1,\dots,m} \in (0, \infty): \prod_{n=1}^m z(n) \leq N\Big\}.
	\end{equation} 
For this infimum, we derive a simple solution exploiting the arithmetic-geometric inequality
using Lemma~\ref{lem_arithm_geom}. 
Setting $\widetilde{z}(n):= \gamma_{N,m}n^{-(H+\half)}$, with $\gamma_{N,m}:= N^{\frac{1}{m}} \Big( \prod_{j=1}^m j^{-(2H+1)}\Big)^{-\frac{1}{2m}},$ $ n=1, \dots, m$, we get
$$
\mathfrak{I} = \sum_{n=1}^m  \frac{1}{n^{2H+1}\widetilde{z}(n)^2}
			= N^{-\frac{2}{m}}m \Big( \prod_{n=1}^m n^{-(2H+1)}\Big)^{\frac{1}{m}},
$$
and notice that the sequence $\{\widetilde{z}(n)\}$ is decreasing.
Since ultimately the vector $\dd$ consists of integers, 
we use $\widetilde{d}(n)=\lfloor \widetilde{z}(n) \rfloor$, $n =1, \dots, m$. 
In fact, this choice guarantees that
$$
\prod_{n=1}^m \widetilde{d}(n) =\prod_{n=1}^m \lfloor \widetilde{z}(n)\rfloor  \leq \prod_{n=1}^m \widetilde{z}(n)=N.
$$	
Furthermore, setting $\widetilde{d}(j)=\lfloor \widetilde{z}(j)\rfloor $
for each $j\in \{1,\dots,m\}$, we obtain
$$
\frac{\widetilde{d}(j)+1}{(j^{-(2H+1)})^{\half}}
=j^{H+\frac{1}{2}}(\lfloor \widetilde{z}(j)\rfloor +1)
			\geq j^{H+\frac{1}{2}}\widetilde{z}(j)
			=\frac{j^{H+\half} N^{\frac{1}{m}}}{j^{H+\half}}\left\{\prod_{n=1}^m \frac{1}{n^{2H+1}}\right\}^{-\frac{1}{2m}}
			=N^{\frac{1}{m}} \left\{ \prod_{n=1}^m \frac{1}{n^{2H+1}}\right\}^{-\frac{1}{2m}}.
$$
Ordering the terms, we have
	$(\widetilde{d}(j)+1)^2 N^{-\frac{2}{m}} \Big( \prod_{n=1}^m n^{-(2H+1)}\Big)^{\frac{1}{m}}\geq j^{-(2H+1)}$,
for each $j \in \{ 1, \dots,m\}$.
From this we deduce the following inequality (notice that the left-hand side term is defined only if $\widetilde{d}(1),\dots,\widetilde{d}(m)>0$):
	\begin{align}\label{secondes}
		\sum_{j=1}^m j^{-(2H+1)}\widetilde{d}(j)^{-2} &
		\leq  \sum_{j=1}^m  \Big(\frac{\widetilde{d}(j)+1}{\widetilde{d}(j)}\Big)^2 N^{-\frac{2}{m}} \Big( \prod_{n=1}^m n^{-(2H+1)}\Big)^{\frac{1}{m}}\\&
		=  N^{-\frac{2}{m}} \Big( \prod_{n=1}^m n^{-(2H+1)}\Big)^{\frac{1}{m}} \sum_{j=1}^m  \Big(\frac{\widetilde{d}(j)+1}{\widetilde{d}(j)}\Big)^2\\&
		\leq  4m N^{-\frac{2}{m}} \Big( \prod_{n=1}^m n^{-(2H+1)}\Big)^{\frac{1}{m}}.
	\end{align}
Hence, we are able to make a first error estimation, placing in the internal minimization of the right-hand side of~\eqref{firstes} the result of inequality in~\eqref{secondes}.
{\small	
	\begin{align}\label{thirdes}
		\mathbb{E}\left[\left\|Z-\widehat{Z}^{N,\star}\right\|_{L^2[0,1]}^2 \right]
		& \leq  \widetilde{C} \inf \left\{ \sum_{k\geq m+1} \frac{1}{k^{2H+1}} + 4m N^{-\frac{2}{m}} \left( \prod_{n=1}^m n^{-(2H+1)}\right)^{\frac{1}{m}}, m \in I(N) \right\}\\
		& \leq C' \inf \left\{ \sum_{k\geq m+1} \frac{1}{k^{2H+1}} + m N^{-\frac{2}{m}} \left(\prod_{n=1}^m n^{-(2H+1)}\right)^{\frac{1}{m}}, m \in I(N) \right\},
	\end{align}
	}
	where $C'=4\widetilde{C}$ and  the set
\begin{equation}\label{eq_def_I(N)}
I(N) :=\{ m \in \NN: N^{\frac{2}{m}} m^{-(2H+1)}\Big( \prod_{n=1}^m n^{-(2H+1)}\Big)^{-\frac{1}{m}}\geq 1 \},
\end{equation}
which represents all $m$'s such that all $\widetilde{d}(1), \dots, \widetilde{d}(m)$ are positive integers. This is to avoid the case where $\prod_{i=1}^m \widetilde{d}(i)\leq N$ holds only because one of the factors is zero.
	In fact, for all $ n \in\{1, \dots,m \}$, $\widetilde{d}(n)=\lfloor \widetilde{z}(n)\rfloor$ is a positive integer if and only  if $\widetilde{z}(n)\geq1 $.	
	 Thanks to  the monotonicity of $\{z(n)\}_{n=1,\dots,m}$,  we only need to check that 
	$$
		\widetilde{z}(m) =  N^{\frac{1}{m}} m^{-(H+\frac{1}{2})}\Big( \prod_{n=1}^m n^{-(2H+1)}\Big)^{-\frac{1}{2m}}\geq 1.
	$$
	First, let us show that $I(N)$, defined in~\eqref{eq_def_I(N)} for each $N\geq1$, is a non-empty finite set with maximum given by $m^*(N)$ of order $\log(N)$.
	We can rewrite it as
	$I(N)=\{ m\geq 1: a_m \leq \log(N) \}$,
	where 
	\begin{equation}
		a_n
		= \frac{1}{2}\log \left(  \prod_{j=1}^n  \frac{n^{2H+1}}{j^{2H+1}} \right).
	\end{equation}
	
	We can now verify that the sequence $a_n$ is increasing in $n\in \NN$: 
	\begin{eqnarray*}
			& & a_n \leq a_{n+1}\\
			&\Longleftrightarrow     & \sum_{j=1}^n \log\left(j^{-(2H+1)}\right)-n\log\left(n^{-(2H+1)}\right) \leq  
\sum_{j=1}^{n+1} \log\left(j^{-(2H+1)}\right) -(n+1)\log\left((n+1)^{-(2H+1)}\right)\\
			&\Longleftrightarrow   & -n\log\left(n^{-(2H+1)}\right) \leq \log\left((n+1)^{-(2H+1)}\right)
 - (n+1)\log\left((n+1)^{-(2H+1)}\right)\\
		    &\Longleftrightarrow   & \log\left(n^{-(2H+1)}\right)
\geq \log\left((n+1)^{-(2H+1)}\right),
	\end{eqnarray*}
which is obviously true. Furthermore the sequence $(a_n)_{n}$ diverges to infinity since
$$
\prod_{j=1}^n \frac{n^{(2H+1)}}{j^{(2H+1)}} 
			= {n^{(2H+1)n}} \prod_{j=1}^n \frac{1}{j^{(2H+1)}} 
			\geq {n^{(2H+1)n}} \prod_{j=2}^n \frac{1}{j^{(2H+1)}} 
			\geq {n^{(2H+1)n}} \frac{1}{n^{(2H+1)(n-1)}}
			\geq {n^{(2H+1)}}.
$$
and $H \in (0, \half)$.
We immediately deduce that $I(N)$ is finite and, since $\{1\}\subset I(N)$, it is also non-empty. 
Hence 	$I(N)=\{1, \dots, m^*(N)\}$. 
Moreover, for all $N \geq 1$, $a_{m^*(N)}\leq \log(N) < a_{m^*(N)+1}$, 
which implies that $m^*(N) = \mathcal{O}(\log(N))$.
	\\
	Now, the error estimation in \eqref{thirdes} can  be further simplified exploiting the fact that, for each $ m \in I(N)$, 
$$
m N^{-\frac{2}{m}} \left( \prod_{n=1}^m n^{-(2H+1)}\right)^{\frac{1}{m}}
= m m^{-(2H+1)}\left(m^{-(2H+1)} N^{\frac{2}{m}} \left( \prod_{n=1}^m n^{-(2H+1)}\right)^{-\frac{1}{m}}\right)^{-1}
\leq m^{-2H}.
$$
The last inequality is a consequence of the fact that $\left( \prod_{n=1}^m n^{-(2H+1)}\right)^{-\frac{1}{m}}\geq 1$ by definition.
	Hence, 
	\begin{equation}\label{fourthes}
			\mathbb{E}\Big[ \|Z-\widehat{Z}^{N,\star}\|_{L^2[0,1]}^2 \Big]
			\leq {C'}\inf \Bigg\{ \sum_{k\geq m+1} \frac{1}{k^{2H+1}}+ m^{-2H} , m \in I(N) \Bigg\},
	\end{equation}
	for some suitable constant $C'>0$.

Consider now the sequence $\{b_n\}_{n \in \NN}$, given by $b_n=\sum_{k\geq n+1} \frac{1}{k^{2H+1}}+ n^{-2H}$. 
For $n\geq 1$,
$$
b_{n+1} - b_n 
= \sum_{k\geq n+2} \frac{1}{k^{2H+1}}+ \frac{1}{(n+1)^{2H}} -\left[\sum_{k\geq n+1} \frac{1}{k^{2H+1}}+ \frac{1}{n^{2H}}\right]
 = - \frac{1}{(n+1)^{2H}} + \frac{1}{(n+1)^{2H+1}}- \frac{1}{n^{2H}}\leq 0,
$$
so that the sequence is decreasing and the infimum in~\eqref{fourthes} is attained at $m=m^*(N)$. 
Therefore,
	\begin{eqnarray*}
			\mathbb{E}\left[ \|Z-\widehat{Z}^{N,\star}\|_{L^2[0,1]}^2 \right]&
			\leq & {C'}\inf \left\{ \sum_{k\geq m+1} \frac{1}{k^{2H+1}}+ m^{-2H} , m \in I(N) \right\}\\
			& = & C'\left( \sum_{k\geq m^*(N)+1} \frac{1}{k^{2H+1}}+ m^*(N)^{-2H} \right)
			 \leq C'\left({m^*(N)^{-2H-1+1}}+ m^*(N)^{-2H} \right)\\
			& = & 2C' m^*(N)^{-2H} \leq C \log(N)^{-2H}.
	\end{eqnarray*}

\subsection{Proof of Remark~\ref{Cont_v_Y}}\label{App_P_Cont_v_Y}
	This can be proved specializing the computations done in~\cite[page 656]{Sharp}. 
	Consider an arbitrary index $n \geq 1.$ For all $t,s \in[0,1]$, exploiting Assumption~\ref{AssumptionAB}, we have that, for any $\rho \in [0,1]$,
	\begin{eqnarray*}
			\left|\mathcal{K}[\psi_n](t)-\mathcal{K}[\psi_n](s)\right|&
			= &\big|\mathcal{K}[\psi_n](t)-\mathcal{K}[\psi_n](s)\big|^\rho\big|\mathcal{K}[\psi_n](t)-\mathcal{K}[\psi_n](s)\big|^{1-\rho}\\&
			\leq & \left( \sup_{u,v \in [0,1], u\neq v }\frac{| \mathcal{K}[\psi_n] (u)- \mathcal{K}[\psi_n] (v) |}{|u-v|^{H+\frac{1}{2}}} |t-s|^{H+\frac{1}{2}} \right)^\rho \left( 2 \sup_{t \in [0,1]}\mathcal{K}[\psi_n](t)\right)^{1-\rho}\\&
			\leq & (C_1n)^\rho(2 C_2 n^{-(H+\frac{1}{2})})^{1-\rho}|t-s|^{\rho(H+\frac{1}{2})}
			=  C_\rho n^{\rho (H+\frac{3}{2})-(H+\frac{1}{2})}|t-s|^{\rho(H+\frac{1}{2})}, 
	\end{eqnarray*}
	where $C_{\rho}:=C_1^\rho(2  C_2)^{1-\rho}<\infty.$ 
	Therefore
	\begin{equation}\label{sop}
		 \left[\mathcal{K}[\psi_n]\right]_{\rho(H+\frac{1}{2})} 
		 = \sup_{t\neq s \in [0,1]}\frac{\left|\mathcal{K}[\psi_n](t)-\mathcal{K}[\psi_n](s)\right|}{|t-s|^{\rho(H+\frac{1}{2})}} 
		 \leq C_\rho n^{\rho (H+\frac{3}{2})-(H+\frac{1}{2})}.
	\end{equation}
	Notice that $\rho (H+\frac{3}{2})-(H+\frac{1}{2}) < -\frac{1}{2}$
when $\rho \in [0,\frac{H}{H+3/2}]$ so that~\eqref{sop} implies
	\begin{equation}
		\sum_{n=1}^\infty  \left[\mathcal{K}[\psi_n]\right]_{\rho(H+\frac{1}{2})} ^2 
		\leq C_\rho^2 \sum_{n=1}^\infty  n^{2\rho (H+\frac{3}{2})-2(H+\frac{1}{2})}
		\leq C_\rho^2 \sum_{n=1}^\infty  n^{-(1+\varepsilon)}
		=K
		< \infty.
	\end{equation}
	In particular, 
$$
\mathbb{E}\left[|Y_t-Y_s|^2\right]
			=\sum_{n=1}^\infty\left|\mathcal{K}[\psi_n](t)-\mathcal{K}[\psi_n](s)\right|^2
			\leq \sum_{n=1}^\infty  \left[\mathcal{K}[\psi_n]\right]_{\rho(H+\frac{1}{2})} ^2 |t-s|^{2\rho(H+\frac{1}{2})}
\leq K |t-s|^{2\rho(H+\frac{1}{2})}.
$$
	\\
As noticed in Remark~\ref{gau} the process $Y$ is centered Gaussian. 
Hence, for each $t, s \in [0,1]$ so is $Y_t-Y_s$. 
Proposition~\ref{Prop_mom_gauss} therefore implies that, for any $r \in \NN$,			
	\begin{equation}
		\mathbb{E}\left[|Y_t-Y_s|^{2r}\right]
		=\mathbb{E}\left[|Y_t-Y_s|^2\right]^{r}(2r-1)!!
		\leq K' |t-s|^{2r\rho(H+\frac{1}{2})},
	\end{equation}
	where $K'= K^{r}(2r-1)!! $, 
yielding existence of a continuous version of~$Y$ since choosing $r\in \NN$ such that $ 2r\rho(H+\frac{1}{2}) >1$, Kolmogorov continuity theorem~\cite[Theorem 3.23]{Kallenberg} applies directly.

\subsection{Proof of Lemma~\ref{lem:Representation}}
\label{App_P_lem:Representation}
Let $\hp:=H+\half$.
Using~\cite[Corollary 1, Equation (12)]{Karp} (with $\psi = b_2 + b_1-a>1/2$), the identity
	$$
	{}_{1}{F}_{2}(a, b_1,b_2, -r)
	 = \frac{\Gamma(b_1)\Gamma(b_2)}{\Gamma(a)\sqrt{\pi}}\int_{0}^{1} G_{2,2}^{2,0}\left([b_1,b_2], \left[a, \half\right], u\right)\cos\left(2\sqrt{ru}\right)\frac{du}{u},
	$$
holds for all $r>0$, where~$G$ denotes the Meijer-G function, generally defined through the so-called Mellin-Barnes type integral~\cite[Equation (1), Section 5.2]{Luke}) as
	\begin{equation}
		G_{p,q}^{\,m,n}\!\left([a_{1},\dots ,a_{p}],[b_{1},\dots ,b_{q}],z\right)
		={\frac {1}{2\pi i}}\int _{L}{\frac {\prod _{j=1}^{m}\Gamma (b_{j}-s)\prod _{j=1}^{n}\Gamma (1-a_{j}+s)}{\prod _{j=m+1}^{q}\Gamma (1-b_{j}+s)\prod _{j=n+1}^{p}\Gamma (a_{j}-s)}}\,z^{s}\,ds.
	\end{equation}
	This representation holds if $z \neq 0$, $0 \leq m \leq q$ and $0 \leq n \leq p$, 
for integers $m, n, p, q$, and $a_k -  b_j \neq 1, 2, 3, \dots$, for $k = 1, 2,\dots, n$ and $j = 1, 2, \dots, m$. 
The last constraint is set to prevent any pole of any $\Gamma(b_j- s), j = 1, 2, \dots, m,$ from  coinciding with any pole of any $\Gamma (1 -a_k + s), k = 1, 2, \dots, n$.
	With $a>0$, $b_2 = 1+a$ and $b_1=\half$,
since $G_{2,2}^{2,0}\left(\left[\half,a+1\right], \left[a, \half\right], u\right) = u^a$, we can therefore write
	\begin{equation}\label{eq:1F2Cos}
		\int_{0}^{1}u^{a-1} \cos\left(2\sqrt{ru}\right)du
		 = \frac{1}{a}{}_{1}{F}_{2}\left(a;\half,a+1; -r\right).
	\end{equation}
Similarly, using integration by parts and properties of generalised Hypergeometric functions,
	\begin{align}\label{eq:1F2Sin}
			\int_{0}^{1}u^{a-1} \sin\left(2\sqrt{ru}\right)du &
			=  \frac{\sin(2\sqrt{r})}{a}-\frac{\sqrt{r}}{a}\int_0^1u^{a-\frac{1}{2}}\cos(2 \sqrt{ru})du\\&
			=  \frac{\sin(2\sqrt{r})}{a}-\frac{\sqrt{r}}{a(a+\half)}{}_{1}{F}_{2}\left(a+\half;\half, a+\frac{3}{2}; -r\right)\\&
			=  \frac{2\sqrt{r}}{a+\half}{}_{1}{F}_{2}\left(a+\half;\frac{3}{2},a+\frac{3}{2}; -r\right),
	\end{align}
	where the last step follows from the definition of generalized sine function $\sin(z)= z \, {}_0 F_{1}(\frac{3}{2},-\frac{1}{4}z^2)$.
Indeed, exploiting~\eqref{eq:Fpq}, we have
	\begin{eqnarray*}
			\frac{\sin(2\sqrt{r})}{a}& - &\frac{\sqrt{r}}{a(a+\half)}{}_{1}{F}_{2}\left(a+\half;\half, a+\frac{3}{2}; -r\right)\\&
			= & \frac{2\sqrt{r}}{a} {}_0 F_{1}\left(\frac{3}{2},-r\right)-\frac{\sqrt{r}}{a(a+\half)}{}_{1}{F}_{2}\left(a+\half;\half, a+\frac{3}{2}; -r\right)\\&
			= & \frac{2\sqrt{r}}{a\left(a+\frac{1}{2}\right)}\left[ \left(a+ \frac{1}{2}\right){}_0 F_{1}\left(\frac{3}{2};-r\right)-\frac{1}{2}{}_{1}{F}_{2}\left(a+\half;\half, a+\frac{3}{2}; -r\right)\right]\\&
			= &\frac{2\sqrt{r}}{a\left(a+\frac{1}{2}\right)}\left[ \left(a+ \frac{1}{2}\right) \sum_{k=0}^{\infty} \frac{(-r)^k}{k!(3/2)_k} - \frac{1}{2} \sum_{k=0}^{\infty}\frac{(a+1/2)_k}{k!(1/2)_k (a+3/2)_k}(-r)^k\right]\\&
			= & \frac{2\sqrt{r}}{a\left(a+\frac{1}{2}\right)}\sum_{k=0}^\infty \frac{1}{k!}\left[  \frac{(a+1/2)}{(3/2)_k}-\frac{1/2(a+1/2)_k}{(1/2)_k (a+3/2)_k}\right](-r)^k\\&
			= & \frac{2\sqrt{r}}{a\left(a+\frac{1}{2}\right)}\sum_{k=0}^\infty \frac{1}{k!}\left[  \frac{a(a+1/2)_k}{(3/2)_k (a+3/2)_k}\right](-r)^k\\&
			= & \frac{2\sqrt{r}}{\left(a+\frac{1}{2}\right)}\sum_{k=0}^\infty \frac{1}{k!} \frac{(a+1/2)_k}{(3/2)_k (a+3/2)_k}(-r)^k
			=  \frac{2\sqrt{r}}{\left(a+\frac{1}{2}\right)}{}_1F_2\left( a+\half;\frac{3}{2},a+\frac{3}{2}; -r\right).
	\end{eqnarray*}
Letting $\alpha:=H-\half$, $\tau:=t-T$, and mapping $v := t-u$, $w:=\frac{v}{t}$ and $y:=w^2$,  we write
	\begin{align}\label{eq:MainIntegral}
		\int_{0}^{T}(t-u)^{\alpha}\E^{\I\pi u}du &
		= \E^{\I\pi t}\int_{(t-T)}^{t}v^{\alpha}\E^{-\I\pi v}dv
		= \E^{\I\pi t}\left[\int_{0}^{t}v^{\alpha}\E^{-\I\pi v}dv - \int_{0}^{\tau}v^{\alpha}\E^{-\I\pi v}dv\right]\nonumber\\&
		= \E^{\I\pi t}\left[t^{1+\alpha} \int_{0}^{1}w^{\alpha}\E^{-\I\pi wt} dw
 - \tau^{1+\alpha}\int_{0}^{1}w^{\alpha}\E^{-\I\pi w\tau} dw\right]\nonumber\\& 
= \frac{\E^{\I\pi t}}{2}\left[t^{1+\alpha}\int_{0}^{1}y^{\frac{\alpha-1}{2}}\E^{-\I\pi t\sqrt{y}} dy - \tau^{1+\alpha}\int_{0}^{1}y^{\frac{\alpha-1}{2}}\E^{-\I\pi y\tau \sqrt{y}}dy\right]\nonumber\\&  
		= \frac{\E^{\I\pi t}}{2}\left[I(t) - I(\tau)	\right],
	\end{align}
	where $I(z):= z^{1+\alpha}\int_{0}^{1}v^{\frac{\alpha-1}{2}}\E^{-\I\pi z \sqrt{v}}dv$.\\
	We  therefore write, for $z\in\{t, \tau\}$,
using~\eqref{eq:1F2Cos}-\eqref{eq:1F2Sin}, $\pi z = 2\sqrt{r}$, and identifying $a-1 = \frac{\alpha-1}{2}$, 
	\begin{align*}
		I(z) & =  z^{1+\alpha}\int_{0}^{1}v^{\frac{\alpha-1}{2}}\E^{-\I\pi z\sqrt{v}} dv
		= z^{1+\alpha}\int_{0}^{1}v^{\frac{\alpha-1}{2}}\cos(\pi z\sqrt{v})dv - \I z^{1+\alpha}\int_{0}^{1}v^{\frac{\alpha-1}{2}}\sin(\pi z\sqrt{v})dv\\&  
		= \frac{2 z^{1+\alpha}}{\hp}{}_{1}{F}_{2}\left(\frac{\hp}{2};\half,1+\frac{\hp}{2}; -r\right) - \I z^{\hp} \frac{4\sqrt{r}}{1+\hp}{}_{1}{F}_{2}\left(\half+\frac{\hp}{2}; \frac{3}{2},\frac{3}{2}+\frac{\hp}{2}; -r\right)\\&  
		= \frac{ z^{\hp}}{h_1}{}_{1}{F}_{2}\left(h_1; \half,1+h_1; -\frac{\pi^2 z^2}{4}\right) - \I \frac{\pi z^{1+\hp}}{h_2}{}_{1}{F}_{2}\left(h_2; \frac{3}{2},1+h_2; -\frac{\pi^2 z^2}{4}\right),
	\end{align*}
	since $\alpha = H-\half = \hp - 1$, $h_1=\frac{\hp}{2}$ and $h_2=\half + h_1$.
	Plugging these into~\eqref{eq:MainIntegral}, we obtain
	{\small\begin{align*}
		\int_{0}^{T}(t-u)^{\alpha}\E^{\I\pi u} du & = 			\frac{\E^{\I\pi t}}{2}\left[I(t) - I(\tau)\right]\\& 
		= \frac{\E^{\I\pi t}}{2}\Bigg[\frac{ z^{\hp}}{h_1}{}_{1}{F}_{2}\left(h_1;\half,1+h_1; -\frac{\pi^2 z^2}{4}\right) - \I \frac{\pi z^{1+\hp}}{h_2}{}_{1}{F}_{2}\left(h_2;\frac{3}{2},1+h_2; -\frac{\pi^2 z^2}{4}\right)\Bigg]_{z=t}\\&
		\quad - \frac{\E^{\I\pi t}}{2}\left[\frac{ z^{\hp}}{h_1}{}_{1}{F}_{2}\left(h_1;\half,1+h_1; -\frac{\pi^2 z^2}{4}\right)- \I \frac{\pi z^{1+\hp}}{h_2}{}_{1}{F}_{2}\left(h_2;\frac{3}{2},1+h_2; -\frac{\pi^2 z^2}{4}\right)\right]_{z=\tau}\\& 
		= \frac{\E^{\I\pi t}}{2h_1}\left[(t)^{\hp}{}_{1}{F}_{2}\left(h_1;\half,1+h_1; -\frac{\pi^2 t^2}{4}\right) - (\tau)^{\hp}{}_{1}{F}_{2}\left(h_1; \half,1+h_1; -\frac{\pi^2 \tau^2}{4}\right)\right]\\&
		\quad -\I\frac{\pi\E^{\I\pi t}}{2h_2}\left[(t)^{1+\hp}{}_{1}{F}_{2}\left(h_2;\frac{3}{2},1+h_2; -\frac{\pi^2  t^2}{4}\right)- (\tau)^{1+\hp}{}_{1}{F}_{2}\left(h_2; \frac{3}{2},1+h_2; -\frac{\pi^2 \tau^2}{4}\right)\right]\\& 
		= {\E^{\I\pi t}}\left[\zeta_{\frac{1}{2}}(t, h_1) - \zeta_{\frac{1}{2}}(\tau, h_1)- \I\pi \left(\zeta_{\frac{3}{2}}(t, h_2) - \zeta_{\frac{3}{2}}(\tau, h_2)\right)\right],
	\end{align*}}
where $\chi(z) := -\frac{1}{4}\pi^2 z^2$ and~$\zeta_{\frac{1}{2}}$ and~$\zeta_{\frac{3}{2}}$ as defined in the lemma.

\subsection{Proof of Lemma~\ref{Lemma_V_H}}\label{App_P_Lemma_V_H}
We first prove~\textbf{(A)}. For each $n \in \NN$ and all $t \in [T,1]$, recall that 
	\[
		\mathcal{K}_H^T[\psi_n](t)
		=\sqrt{2}\int_0^T (t-u)^{H-\frac{1}{2}} \cos\left(\frac{u}{\sqrt{\lambda_n}}\right)du
 =\sqrt{2}\int_{t-T}^t v^{H-\frac{1}{2}} \cos\left(\frac{t-v}{\sqrt{\lambda_n}}\right)dv,
	\]
with the change  of variables $v=t-u$.
Assume ${T} \leq s < t \leq 1$. Two situations are possible:
	\begin{itemize}
		\item If $0 \leq s-T < {t-T} \leq s < t \leq 1$, we have
		\begin{eqnarray*}
			\left| \mathcal{K}_H^T[\psi_n](t)-\mathcal{K}_H^T[\psi_n](s) \right| &
			= & \sqrt{2}\left| \int_{t-T}^t v^{H-\frac{1}{2}} 	\cos\left(\frac{t-v}{\sqrt{\lambda_n}}\right)dv-\int_{s-T}^s v^{H-\frac{1}{2}} \cos\left(\frac{s-v}{\sqrt{\lambda_n}}\right)dv\right| \\
			& \leq & \sqrt{2}\Bigg(\left| \int_{t-T}^s v^{H-\frac{1}{2}} \left(\cos\left(\frac{t-v}{\sqrt{\lambda_n}}\right)-\cos\left(\frac{s-v}{\sqrt{\lambda_n}}\right)\right)dv\right| \\
			& \qquad & + \left| \int_{s}^t v^{H-\frac{1}{2}} \cos\left(\frac{t-v}{\sqrt{\lambda_n}}\right)dv\right| 
+ \left| \int_{s-T}^{t-T} v^{H-\frac{1}{2}} \cos\left(\frac{s-v}{\sqrt{\lambda_n}}\right)dv\right| \Bigg) \\
			& \leq & \sqrt{2}\Bigg( \int_{t-T}^s v^{H-\frac{1}{2}} \left|\cos\left(\frac{t-v}{\sqrt{\lambda_n}}\right)- \cos\left(\frac{s-v}{\sqrt{\lambda_n}}\right)\right|dv\\
			& \qquad & + \int_{s}^t v^{H-\frac{1}{2}} dv + \int_{s-T}^{t-T} v^{H-\frac{1}{2}} dv \Bigg) \\
			& \leq &  \sqrt{2}\Bigg( \int_{t-T}^s v^{H-\frac{1}{2}} \left|\frac{t-s}{\sqrt{\lambda_n}}\right|dv + K |t-s|^{H+\frac{1}{2}} + K |t-s|^{H+\frac{1}{2}}\Bigg) \\
			& \leq & \sqrt{2}\Bigg(\frac{|t-s|}{\sqrt{\lambda_n}} \int_{t-T}^s v^{H-\frac{1}{2}} dv +2 K |t-s|^{H+\frac{1}{2}} \Bigg) \\
			& \leq & \sqrt{2}\Bigg(\frac{|t-s|}{\sqrt{\lambda_n}} \| (\cdot)^{H-\frac{1}{2}}\|_{L^1[0,1]} +2 K |t-s|^{H+\frac{1}{2}} \Bigg) 
			\leq \widetilde{C}_1^T |t-s|^{H+\frac{1}{2}},
		\end{eqnarray*}
with
$ \widetilde{C}_1^T=\max\left\{2\sqrt{2}K,\sqrt{ \frac{{2}}{\lambda_n}} \| (\cdot)^{H-\frac{1}{2}}\|_{L^1[0,1]}\right\}=\max\left\{2\sqrt{2}K, \frac{\sqrt{2}(2n-1)\pi}{2} \| (\cdot)^{H-\frac{1}{2}}\|_{L^1[0,1]}\right\}$,
since $\cos(\cdot)$ is Lipschitz on any compact and $\int_{0}^{\cdot} v^{H-\frac{1}{2}} dv$ is $\left(H+\frac{1}{2}\right)$-H\"{o}lder continuous.

		\item If $0 \leq s-T \leq s \leq t-T \leq t \leq 1$, 
		\begin{eqnarray*}
			\left| \mathcal{K}_H^T[\psi_n](t)-\mathcal{K}_H^T[\psi_n](s) \right|
			& = & \sqrt{2}\left| \int_{t-T}^t v^{H-\frac{1}{2}} \cos\left(\frac{t-v}{\sqrt{\lambda_n}}\right)dv-\int_{s-T}^s v^{H-\frac{1}{2}} \cos\left(\frac{s-v}{\sqrt{\lambda_n}}\right)dv\right| \\
			& = & \sqrt{2}\Bigg| \int_{t-T}^t v^{H-\frac{1}{2}} \cos\left(\frac{t-v}{\sqrt{\lambda_n}}\right)dv-\int_{s-T}^s v^{H-\frac{1}{2}} \cos\left(\frac{s-v}{\sqrt{\lambda_n}}\right)dv \\
			& \qquad & + \int_s^{t-T}v^{H-\frac{1}{2}}\cos\left(\frac{t-v}{\sqrt{\lambda_n}} \right) dv -  \int_s^{t-T}v^{H-\frac{1}{2}}\cos\left(\frac{t-v}{\sqrt{\lambda_n}} \right) dv \\
& & +  \int_s^{t-T}v^{H-\frac{1}{2}}\cos\left(\frac{s-v}{\sqrt{\lambda_n}} \right) dv -  \int_s^{t-T}v^{H-\frac{1}{2}}\cos\left(\frac{s-v}{\sqrt{\lambda_n}} \right) dv  \Bigg| \\
			& \leq & \sqrt{2}\Bigg(\left| \int_s^{t-T} v^{H-\frac{1}{2}} \left(\cos\left(\frac{t-v}{\sqrt{\lambda_n}}\right)-\cos\left(\frac{s-v}{\sqrt{\lambda_n}}\right)\right)dv\right| \\
			& \qquad & + \left| \int_{s}^t v^{H-\frac{1}{2}} \cos\left(\frac{t-v}{\sqrt{\lambda_n}}\right)dv\right| 
+ \left| \int_{s-T}^{t-T} v^{H-\frac{1}{2}} \cos\left(\frac{s-v}{\sqrt{\lambda_n}}\right)dv\right| \Bigg) 	\\
			& \leq &  \dots 
			\leq  \widetilde{C}_1^T |t-s|^{H+\frac{1}{2}},
		\end{eqnarray*}
		where the dots correspond to the same computations as in the previous case  and leads to the same estimation with the same constant~$\widetilde{C}_1^T$.
	\end{itemize}
	This proves \textbf{(A)}.

To prove \textbf{(B)}, recall that, for $T \in [0,1]$ and $n\in\NN$,  the function $ \mathcal{K}^T_H[ \psi_n] :[T,1]\to \R$ reads
\begin{align}\label{eq:main}
\mathcal{K}^T_H [\psi_n](t) 
 & = \sqrt{2}\int_{0}^{T}(t-s)^{H-\half}\cos\left(\left(n-\half\right)\pi s\right) ds\nonumber\\
 &  = \frac{\sqrt{2}}{m^{H+\half}}\int_{0}^{mT}(mt-u)^{H-\half}\cos\left(\pi u\right) du 
		=: \Phi_m(t).
\end{align}
with the change of variable $u=(n-\half)s=:ms$.
	Denote from now on $\NNt:=\{m = n-\half, n\in\NN\}$.
	From~\eqref{eq:main}, we deduce, for each $m\in\NNt$ and $t\in[T,1]$,
	\begin{equation}\label{eq_def_phi_little}
		m^{H+\half}\Phi_m(t) 
		= \sqrt{2}\int_{0}^{mT}(mt-u)^{H-\half}\cos\left(\pi u\right) du
		=:\sqrt{2}\phi_m(t).
	\end{equation}
	To end the proof of \textbf{(B)}, it therefore suffices to show that $(\phi_m(t))_{m\in\NNt, t\in[T,1]}$ is uniformly bounded since, in that case we have
	\begin{eqnarray*}
		\| \mathcal{K}_H^T[\psi_n]\|_{\infty} 
		& = & \sup_{t \in [T,1]}| \mathcal{K}_H^T[\psi_n](t) |
		= \sup_{t \in [T,1]}| \Phi_{n-\frac{1}{2}}(t)|
		= \frac{\sqrt{2}}{(n-\frac{1}{2})^{H+\frac{1}{2}}}\sup_{t \in [T,1]}| \phi_{n-\frac{1}{2}}(t)| \\
		& \leq & \frac{\sqrt{2}}{(n-\frac{1}{2})^{H+\frac{1}{2}}}\sup_{t \in [T,1], m \in \widetilde{\NN}}| \phi_{m}(t)|
		\leq  \frac{\sqrt{2}}{(n-\frac{1}{2})^{H+\frac{1}{2}}} C
		\leq C_2^T n ^{-(H+\frac{1}{2})},
	\end{eqnarray*}
	for some $C_2^T>0$, proving~\textbf{(B)}.
	The following guarantees the uniform boundedness of~$\phi_x$ in~\eqref{eq_def_phi_little}.

	\begin{proposition}
		For any $T\in [0,1]$, there exists $C>0$
such that $|\phi_x(t)| \leq C$ for all $x\geq 0$, $t \in [T,1]$.
	\end{proposition}
	\begin{proof}
		For $x>0$, we write
		$$
		\phi_x(t) 
		= \int_{0}^{xT}(xt-u)^{H-\half}\cos\left(\pi u\right) du
		= \Re\left\{\int_{0}^{xT}(xt-u)^{H-\half}\E^{\I\pi u} du\right\}.
		$$
		Using the representation in Lemma~\ref{lem:Representation}, we are thus left to prove that the maps $\zeta_{\frac{1}{2}}(\cdot, h_1)$ and $\zeta_{\frac{3}{2}}(\cdot, h_2)$, defined in~\eqref{def_zeta_khz}, are bounded on $[0,\infty)$ by, say $L_{\frac{1}{2}}$ and $L_{\frac{3}{2}}$. Indeed, in this case,
		{\footnotesize
		\begin{eqnarray*}
			\sup_{x>0, t \in [T,1]}|\phi_x(t)|
			& = & \sup_{x>0, t \in [T,1]}\left| \int_0^{xT}(xt-u)^{H-\frac{1}{2}} \E^{\I \pi u} du \right|\\
			& \leq & \sup_{x>0, t \in [T,1]}\Bigg| \frac{\E^{\I\pi xt}}{2}\Bigg[\Big(\zeta_{\frac{1}{2}}(xt, h_1) - \zeta_{\frac{1}{2}}(x(t-T), h_1)\Big)
- \I\pi \Big(\zeta_{\frac{3}{2}}(xt, h_2) - \zeta_{\frac{3}{2}}(x(t-T), h_2)\Big)\Bigg] \Bigg|\\
			& \leq & \half \sup_{y,z \in [0, \infty)}\Bigg|\Big(\zeta_{\frac{1}{2}}(y, h_1) - \zeta_{\frac{1}{2}}(z, h_1)\Big)
- \I\pi \Big(\zeta_{\frac{3}{2}}(y, h_2) - \zeta_{\frac{3}{2}}(z, h_2)\Big)\Bigg|\\
			& \leq & \pi  \left\{\sup_{y \in [0, \infty)}\left| \zeta_{\frac{1}{2}}(y, h_1)\right| +\sup_{y \in [0, \infty)} \left| \zeta_{\frac{3}{2}}(y, h_2) \right|\right\}\leq L_{\frac{1}{2}}+L_{\frac{3}{2}} =C < +\infty.
		\end{eqnarray*}}
		The maps $\zeta_{\frac{1}{2}}(\cdot, h_1)$ and $\zeta_{\frac{3}{2}}(\cdot, h_2)$ are both clearly continuous. Moreover, as $z$ tends to infinity $\zeta_{k}(z, h)$ converges to a constant $c_{k}$, for $(k,h)\in(\{\half, \frac{3}{2}\}, \{h_1,h_2\})$. 
The identities
$$
\frac{{}_1 F_{2}\left(h;\half,1+h;-x\right)}{h}
			= \int_0^1 \frac{\cos(2\sqrt{x u})}{u^{1-h}}du
\quad\text{and}\quad
\frac{{}_1 F_{2}\left(h;\frac{3}{2},1+h;-x\right)}{h}
= \frac{1}{2 \sqrt{x}} \int_0^1 \frac{\sin(2\sqrt{x u})}{u^{3/2 - h}} du
$$
hold (this can be checked with \texttt{Wolfram Mathematica} for example) and therefore,
		\begin{eqnarray*}
			\zeta_{\frac{1}{2}}(z,h_1)
			& = & \frac{z^{2 h_1}}{2 h_1} {}_1 F_{2}\left(h_1;\half,1+h_1;-\frac{\pi^2z^2}{4}\right)
			= \frac{z^{2 h_1}}{2} \int_0^1 u^{h_1 -1} \cos( \pi z \sqrt{u}) du\\
			& = & \frac{z^{2 h_1}}{2} \int_0^{\pi z} \frac{x^{2(h_1 -1)}}{(\pi z)^{2(h_1 -1)}} \cos( x)\frac{2x}{\pi^2 z^2} dx
			= \frac{1}{\pi^{2 h_1}} \int_0^{\pi z} x^{2h_1 -1} \cos( x) dx,
		\end{eqnarray*}
		where, in the second line, we used the change of variables $x=\pi z \sqrt{u}$. 
		In particular, as $z$ tends to infinity, this converges to
$ \pi^{-2 h_1} \int_0^{+\infty} x^{2h_1 -1} \cos( x) dx = \frac{\cos(\pi h_1)}{\pi^{2 h_1}}\Gamma(2h_1) =: c_{1/2}\approx 0.440433$.
		\\
		Analogously, for $k=\frac{3}{2}$, 
		\begin{eqnarray*}
			\zeta_{\frac{3}{2}}(z,h_2)
			& = & \frac{z^{2 h_2}}{2 h_2} {}_1 F_{2}\left(h_2;\frac{3}{2},1+h_2;-\frac{\pi^2z^2}{4}\right)
			= \frac{z^{2 h_2}}{2 \pi z} \int_0^1 u^{h_2 -3/2} \sin( \pi z \sqrt{u}) du\\
			& = & \frac{z^{2 h_2-1}}{2 \pi} \int_0^{\pi z} \frac{x^{2h_2 -3)}}{(\pi z)^{2h_2 -3)}} \sin( x)\frac{2x}{\pi^2 z^2} dx 
			= \frac{1}{\pi^{2 h_2}} \int_0^{\pi z} x^{2(h_2 -1)} \sin( x) dx,
		\end{eqnarray*}
		with the same change of variables as before.
This converges to
$ \pi^{-2 h_2} \int_0^{+\infty} x^{2h_2 -2} \sin( x) dx = \frac{-\cos(\pi h_2)}{\pi^{2 h_2}}\Gamma(2h_2-1) =: c_{3/2}\approx 0.193$ as~$z$ tends to infinity. 
		For $k>0$, $\zeta_{k}(z,h) = z^{2h}(1 + \Oo(z^2))$ at zero. 
		Since $H \in (0,\half)$, the two functions are continuous and bounded
and the proposition follows.
	\end{proof}

\subsection{Proof of Theorem~\ref{Thm_pricing_err}}
\label{App_Thm_pricing_err}
We only provide the proof of~\eqref{ineq1_Thm_pricing_err} since, as already noticed, 
that of~\eqref{ineq2_Thm_pricing_err} follows immediately.
Suppose that $F: \R \to \R$ is Lipschitz continuous with constant~$M$. 
By Definitions~\eqref{Eq:VIX^2_dyn} and~\eqref{quant_price_VIX}, we have
\begin{eqnarray*}
		& & \Big|  \mathbb{E}\left[F\left(\VIX_T\right)\right]
 - \mathbb{E}\left[F\left(\widehat{\VIX}_T^{\dd}\right)\right]\Big| \\
		& = & \Bigg| \mathbb{E}\left[F\left(\left|\frac{1}{\Delta}\int_T^{T + \Delta}v_0(t)
\exp\left\{\gamma Z^{T,\Delta}_t+\frac{\gamma^2}{2}\left(\int_0^{t-T}K(s)^2ds - \int_0^{t}K(s)^2ds\right)\right\} dt\right|^{\half}\right) \right] \\
		&\qquad &-\mathbb{E}\left[F \left(\left|\frac{1}{\Delta}\int_T^{T + \Delta}v_0(t)\exp\left\{\gamma\widehat Z^{T,\Delta,\dd}_t + \frac{\gamma^2}{2}\left(\int_0^{t-T}K(s)^2ds - \int_0^{t}K(s)^2ds\right)\right\}dt\right|^{\half}\right)\right] \Bigg|.
\end{eqnarray*}
For clarity, let
$Z := Z^{T,\Delta}$, $\widehat{Z} := \widehat{Z}^{T,\Delta,\dd}$,
$\Hf := \int_T^{T+\Delta}h(t)e^{\gamma Z_t}dt$ and 
$\Hft := \int_T^{T+\Delta}h(t)e^{\gamma \widehat{Z}_t}dt$,
with
$$
h(t):= \frac{v_0(t)}{\Delta} \exp\left\{\frac{\gamma^2}{2}\left(\int_0^{t-T}K(s)^2ds - \int_0^{t}K(s)^2ds\right)\right\}, \qquad \text{for }t \in [T, T+\Delta].
$$
We can therefore write, using the Lipschitz property of~$F$ (with constant~$M$) and Lemma~\ref{lem_local_lip}, 
\begin{align*}
	\left|\mathbb{E}\left[F\left(\VIX_T\right)\right]  - \mathbb{E}\left[F\left(\widehat{\VIX}_T^{\dd}\right)\right]\right|
	& 	= \left|\mathbb{E}\left[F\left(\Hf^{\half}\right)\right]-\mathbb{E}\left[F\left(\Hft^{\half}\right)\right]\right|
	\leq \mathbb{E}\left[\left|F\left(\Hf^{\half}\right)- F\left(\Hft^{\half}\right)\right|\right]\\
	& 	\leq M \mathbb{E}\left[\left|\Hf^{\half}-\Hft^{\half}\right|\right]
	\leq M \mathbb{E}\Bigg[\left(\frac{1}{\Hf}+\frac{1}{\Hft}\right)\left|\Hf-\Hft\right|\Bigg]\\
	& =: M \mathbb{E}\left[A \left|\Hf-\Hft\right|\right] 
	\leq M \mathbb{E}\Bigg[A \int_T^{T+\Delta}h(t)\left|e^{\gamma Z_t}-e^{\gamma  \widehat{Z}_t}\right|dt \Bigg]\\&
	\leq M \mathbb{E}\Bigg[A \int_T^{T+\Delta}h(t)\gamma 
\left(e^{\gamma Z_t}+e^{\gamma \widehat{Z}_t}\right)\left|Z_t- \widehat{Z}_t\right|dt \Bigg].
\end{align*}
Now, an application of H\"{o}lder's inequality yields
\begin{align}
	\left|\mathbb{E}\left[F\left(\VIX_T\right)\right]  - \mathbb{E}\left[F\left(\widehat{\VIX}_T^{\dd}\right)\right]\right|
&
	\leq M \mathbb{E}\left[\gamma A \left|\int_T^{T+\Delta}h(t)^2
\left(e^{\gamma Z_t}+e^{\gamma \widehat{Z}_t}\right)^2 dt\right|^{\half}\left|\int_T^{T+\Delta}\left|Z_t- \widehat{Z}_t\right|^2dt\right|^{\half} \right]\\&
	\leq M \mathbb{E}\left[(\gamma A)^2 \int_T^{T+\Delta}h(t)^2
\left(e^{\gamma Z_t}+e^{\gamma \widehat{Z}_t}\right)^2 dt\right]^{\half}\mathbb{E}\left[\int_T^{T+\Delta}\left|Z_t- \widehat{Z}_t\right|^2dt \right]^{\half}\\&
	= \mathfrak{K} \ \mathbb{E}\left[\int_T^{T+\Delta}\left|Z_t- \widehat{Z}_t\right|^2dt \right]^{\half},
\end{align}
where
$\mathfrak{K}:= M \mathbb{E}[\gamma^2 A^2 
\int_T^{T+\Delta}h(t)^2(e^{\gamma Z_t}+ e^{\gamma  \widehat{Z}_t})^2 dt]^{\half}
$.
It remains to show that $\mathfrak{K}$ is a strictly positive finite constant. 
This follows from the fact that $\{Z_t\}_{t \in [T, T+\Delta]}$ does not explode in finite time (and so does not its quantization $\widehat{Z}$ either).
The identity $(a+b)^2\leq 2(a^2+b^2)$ and H\"{o}lder's inequality imply
\begin{eqnarray*}
	\mathfrak{K}^2 
	& \leq & 4M^2 \gamma^2\mathbb{E}\left[\left(\frac{1}{\Hf}+\frac{1}{\Hft}\right)\int_T^{T+\Delta}h(t)^2
\left(e^{2\gamma Z_t}+ e^{2\gamma  \widehat{Z}_t}\right) dt\right]\\
	& \leq & 4M^2\gamma^2\mathbb{E}\left[\left|\frac{1}{\Hf}+\frac{1}{\Hft}\right|^2\right]^{\half}
\mathbb{E}\left[\left|\int_T^{T+\Delta}h(t)^2
\left(e^{2\gamma Z_t}+ e^{2\gamma  \widehat{Z}_t}\right) dt\right|^2\right]^{\half}\\
	& \leq & 16M^2\gamma^2\mathbb{E}\left[\frac{1}{\Hf^2} + \frac{1}{\Hft^2}\right]^{\half}
\mathbb{E}\left[\left|\int_T^{T+\Delta}h(t)^2e^{2\gamma Z_t}dt\right|^2 + \left|\int_T^{T+\Delta}h(t)^2e^{2\gamma  \widehat{Z}_t} dt\right|^2\right]^{\half}\\
	& = & : 16M^2\gamma^2 (A_1 +A_2)^{\half} (B_1 +B_2)^{\half}.
\end{eqnarray*}
We only need to show that $A_1, A_2, B_1$ and $B_2$ are finite. 
Since~$h$ is a positive continuous function on the compact interval $[T,T+\Delta]$, we have
\begin{align}\label{eq_pr_VV_h}
	\Hf 
	& \geq \int_T^{T+\Delta}\inf_{s \in [T,T+\Delta]}\left(h(s)e^{\gamma Z_s}\right) dt 
	\geq \Delta\inf_{s \in [T,T+\Delta]} h(s)e^{\gamma Z_s} \\
	& \geq  \Delta\inf_{t \in [T,T+\Delta]}h(t)\inf_{s \in [T,T+\Delta]}e^{\gamma Z_s}
	\geq\Delta \widetilde{h}\exp\left\{\gamma \inf_{s \in [T,T+\Delta]}Z_s\right\},
\end{align}
with $\widetilde{h}:= \inf_{t \in [T,T+\Delta]}h(t) > 0$.
The inequality~\eqref{eq_pr_VV_h} implies 
\begin{eqnarray*}
	A_1  
	& = & \mathbb{E}\left[ \Hf^{-2} \right]
	\leq \frac{\mathbb{E}\left[\exp\left\{-2\gamma \inf_{s \in [T,T+\Delta]}Z_s\right\}\right]}{\Delta^2 \widetilde{h}^2}
	= \frac{\mathbb{E}\left[\exp\left\{2\gamma \sup_{s \in [T,T+\Delta]}(-Z_s)\right\}\right]}{\Delta^2 \widetilde{h}^2}
\\
	& = & \frac{1}{\Delta^2 \widetilde{h}^2} \mathbb{E}\left[\exp\left\{2\gamma \sup_{s \in [T,T+\Delta]}Z_s\right\}\right],
\end{eqnarray*}
since $-Z$ and $Z$ have the same law.
The process $Z=(Z_t)_{t \in [T,T+\Delta]}$ is a continuous centered Gaussian process defined on a compact set. Thus, by Theorem 1.5.4 in~\cite{Adler_1}, it is almost surely bounded there. 
Furthermore, exploiting Lemma~\ref{lem_exp_pos_rv} and \emph{Borel-TIS} inequality~\cite[Theorem 2.1.1]{Adler_1}, we have
\begin{align}\label{eq_sup_w_BTIS}
	\nonumber& \mathbb{E}  \left[ e^{2\gamma \sup_{s \in [T,T+\Delta]}Z_s}\right]
	 = : \mathbb{E}\left[ e^{2\gamma \|Z\|}\right]
	=\int_0^{+\infty}\mathbb{P}\left( e^{2\gamma \|Z\|}>u\right) du
	=\int_0^{+\infty}\mathbb{P}\left( \|Z\|> \frac{\log(u)}{2\gamma }\right) du\\
	\nonumber& = \int_0^{e^{2\gamma \mathbb{E}[\|Z\|]}}du+\int_{e^{2\gamma \mathbb{E}[\|V\|]}}^{+\infty}\mathbb{P}\left( \|Z\|> \frac{\log(u)}{2\gamma }\right) du
	= e^{2\gamma \mathbb{E}[\|Z\|]}+\int_{e^{2\gamma \mathbb{E}[\|V\|]}}^{+\infty} e^{-\frac{1}{2}\left( \frac{\frac{1}{2\gamma }\log(u)-\mathbb{E}[\|Z\|]}{\sigma_T}\right)^2} du\\
	& \leq  e^{2\gamma \mathbb{E}[\|Z\|]}
 + \int_{0}^{+\infty} e^{-\frac{1}{2}\left( \frac{\frac{1}{2\gamma}\log(u)-\mathbb{E}[\|Z\|]}{\sigma_T}\right)^2} du,
\end{align}
with $\|Z\|:= \sup_{s \in [T,T+\Delta]}Z_s$ and $\sigma_T^2:= \sup_{t \in [T,T+\Delta]}\mathbb{E}[Z_t^2]$. 
The change of variable $\frac{\log(u)}{2\gamma }=v$ in the last term in~\eqref{eq_sup_w_BTIS} yields 
\[
	\int_{0}^{+\infty} e^{-\frac{1}{2}\left( \frac{\frac{1}{2\gamma}\log(u)-\mathbb{E}[\|Z\|]}{\sigma_T}\right)^2} du
	=  2\gamma\int_{\R}e^{-\frac{1}{2}\left( \frac{v-\mathbb{E}[\|Z\|]}{\sigma_T}\right)^2} e^{2\gamma v}dv
	= \sqrt{2\pi}2\gamma  \mathbb{E}[e^{2\gamma Y}],
\]
since $Y \sim \mathcal{N}(\mathbb{E}[\|Z\|], \sigma_T)$, and hence $A_1$ is finite. 
Now, notice that, in analogy to the last line of the proof of Proposition~\ref{Prop_staz_quant}, for any $t \in [T,T+\Delta]$, we have
\begin{align}\label{Eq:staz_over_t}
	\mathbb{E}\left[Z_t \Big| (\widehat{Z}_s)_{s \in [T, T+\Delta]}\right]&
	= \mathbb{E}\left[\mathbb{E}\left[Z_t \Big| \{\widehat{\xi}_n^{d(n)}\}_{n=1,\dots,m}\right] \Big| (\widehat{Z}_s)_{s \in [T, T+\Delta]}\right]
	= \mathbb{E}\left[\widehat{Z}_t \Big| (\widehat{Z}_s)_{s \in [T, T+\Delta]}\right]
	= \widehat{Z}_t,
\end{align}
since the sigma-algebra generated by $(\widehat{Z}_s)_{s \in [T, T+\Delta]}$ is included in the sigma-algebra generated by $\{\widehat{\xi}_n^{d(n)}\}_{n=1,\dots,m}$.
Now, exploiting, in sequence, \eqref{Eq:staz_over_t}, the conditional version of $\sup_{t \in [T_1,T_2]} \mathbb{E}[f_t] \leq \mathbb{E}[\sup_{t \in [T_1,T_2]} f_t]$, conditional Jensen's inequality together with the convexity of $x\mapsto e^{\gamma x}$, for $\gamma > 0$
 and the tower property, we obtain
\begin{align}
	\mathbb{E}\left[ \exp \left\{\gamma \sup_{t \in [T, T+\Delta]}\widehat{Z}_t \right\}\right]& \nonumber
	=\mathbb{E}\left[ \exp \left\{\gamma \sup_{t \in [T, T+\Delta]}\mathbb{E}\left[Z_t \Big| (\widehat{Z}_s)_{s \in [T, T+\Delta]}\right] \right\}\right]\\
	&\nonumber \leq \mathbb{E}\left[ \exp \left\{\gamma \mathbb{E}\left[\sup_{t \in [T, T+\Delta]} Z_t \Big| (\widehat{Z}_s)_{s \in [T, T+\Delta]}\right] \right\}\right]\\
	&\nonumber \leq \mathbb{E}\left[ \mathbb{E}\left[\exp \left\{\gamma \sup_{t \in [T, T+\Delta]} Z_t  \right\}\Big| (\widehat{Z}_s)_{s \in [T, T+\Delta]}\right]\right]\\
	&	= \mathbb{E}\left[\exp \left\{\gamma \sup_{t \in [T, T+\Delta]} Z_t  \right\}\right].
\end{align}
Thus, we have
\begin{equation*}
	A_2 
	= \mathbb{E}\left[\widehat{\mathfrak{H}}^{-2}\right]
	\leq \frac{1}{\Delta^2 \widetilde{h}^2}\mathbb{E}\left[ \exp \left\{\gamma \sup_{t \in [T, T+\Delta]}\widehat{Z}_t \right\}\right]
	\leq \frac{1}{\Delta^2 \widetilde{h}^2}\mathbb{E}\left[\exp \left\{\gamma \sup_{t \in [T, T+\Delta]} Z_t  \right\}\right],
\end{equation*}
which is finite because of the proof of the finiteness of $A_1$, above.

Exploiting Fubini's theorem we rewrite $B_1$ as
\begin{equation}
	B_1  = \mathbb{E}\left[\left(\int_T^{T+\Delta}h(t)^2 e^{2\gamma  Z_t}dt\right)^2\right]
	= \int_T^{T+\Delta}\int_T^{T+\Delta}h(t)^2 h(s)^2 \mathbb{E}\left[e^{2\gamma ( Z_t +Z_s)}\right] dt ds.
\end{equation}
Since $(Z_t)_{t \in [T,T+\Delta]}$ is centered Gaussian with covariance $\mathbb{E}[Z_tZ_s]=\int_0^T K(t-u)K(s-u) du$, then 
$(Z_t+Z_s) \sim \mathcal{N}(0, g(t,s))$, with $g(t,s):= \mathbb{E}[(Z_t+Z_s)^2]=\int_0^T (K(t-u) + K(s-u))^2 du $ and therefore
\begin{equation}
	B_1  =\int_T^{T+\Delta}\int_T^{T+\Delta}h(t)^2 h(s)^2 e^{2\gamma^2 g(t,s)} dt ds
\end{equation}
is finite since both~$h$ and~$g$ are continuous on  compact intervals.
Finally, for $B_2$ we have

\begin{eqnarray*}
	B_2
	& = &\mathbb{E}\left[\left(\int_T^{T+\Delta}h(t)^2 e^{2\gamma  \widehat{Z}_t}dt\right)^2\right]
	= \int_T^{T+\Delta}\int_T^{T+\Delta}h(t)^2 h(s)^2 \mathbb{E}\left[e^{2\gamma ( \widehat{Z}_t + \widehat{Z}_s)}\right] dt ds\\
	& \leq & \int_T^{T+\Delta}\int_T^{T+\Delta}h(t)^2 h(s)^2 \mathbb{E}\left[e^{2\gamma ( {Z}_t + {Z}_s)}\right] dt ds
	= B_1,
\end{eqnarray*}
where we have used the fact that for all $t,s \in [T,T+\Delta]$, $(\widehat{Z}_t +\widehat{Z}_s)$ is a stationary quantizer for $(Z_t+Z_s)$ and so $\mathbb{E}[e^{2\gamma (\widehat{Z}_t +\widehat{Z}_s)}] \leq \mathbb{E}[e^{2\gamma ({Z}_t +{Z}_s)}]$ since $f(x)=e^{2\gamma x}$ is a convex function (see Remark~\ref{rem_staz_quant} in Section~\ref{Sect_staz}). 
Therefore $B_2$ is finite and the proof follows.

\section{Some useful results}\label{Some_us_res}
We recall some important results used throughout the text. Straightforward proofs are omitted.

\begin{proposition}\label{Prop_mom_gauss}
	For a Gaussian random variable $Z\sim \mathcal{N}(\mu, \sigma)$,
	\begin{equation*}
		\mathbb{E}\left[|Z-\mu|^p\right]
		=\left\{ 
\begin{array}{ll}
(p-1)!!\sigma^p, & \text{if } p\text{ is even},\\
0, & \text{if } p\text{ is odd.}
\end{array}
\right.
	\end{equation*}
\end{proposition}



We recall~\cite[Problem 8.5]{Steele}, correcting a small error, 
used in the proof of Proposition~\ref{prop:choice_m_d}:
\begin{lemma}\label{lem_arithm_geom}
Let $m, N \in \NN$ and $p_1,\dots,p_m$ positive real numbers. 
Then
	\[
		\inf\left\{ \sum_{n=1}^m\frac{p_n}{x_n^2}: \quad x_1, \dots,x_m \in (0,\infty),\quad\prod_{n=1}^mx_n \leq N  \right\}
 = m N^{-\frac{2}{m}} \left( \prod_{j=1}^m p_j \right)^{\frac{1}{m}},
	\]
where the infimum is attained for 
$x_n=N^{\frac{1}{m}}p_n^{\frac{1}{2}}\left( \prod_{j=1}^m p_j \right)^{-\frac{1}{2m}}$, for all $n \in \{1,\dots,m\}$. 
\end{lemma}

\begin{proof}
The general arithmetic-geometric inequalities imply
	\[
		\frac{1}{m}\sum_{n=1}^m \frac{p_n}{x_n^2} 
		\geq \left( \prod_{n=1}^m \frac{p_n}{x_n^2}\right)^{\frac{1}{m}}	
		=\left( \prod_{n=1}^m p_n\right)^{\frac{1}{m}} \left( \prod_{n=1}^m  \frac{1}{x_n^2}\right)^{\frac{1}{m}}
		\geq \left( \prod_{n=1}^m p_n\right)^{\frac{1}{m}} N^{-\frac{2}{m}},
	\]
since $\prod_{n=1}^m x_n \geq N$ by assumption.
The right-hand side does not depend on $x_1,\dots,x_m$, so
	\[
\inf\left\{ \sum_{n=1}^m\frac{p_n}{x_n^2}: \quad x_1, \dots,x_m \in (0,\infty),\quad\prod_{n=1}^mx_n \leq N  \right\} \geq  m \left( \prod_{n=1}^m p_n\right)^{\frac{1}{m}} N^{-\frac{2}{m}}.
	\]
Choosing $\widetilde{x}_n=N^{\frac{1}{m}}p_n^{\frac{1}{2}}\left( \prod_{j=1}^m p_j \right)^{-\frac{1}{2m}}$,  for all $n \in \{1,\dots,m\}$, we obtain
	\begin{eqnarray*}
		m \left( \prod_{n=1}^m \frac{ p_n}{N^2}\right)^{\frac{1}{m}} 
		=\sum_{n=1}^m\frac{p_n}{\widetilde{x}_n^2}
		\geq\inf\left\{ \sum_{n=1}^m\frac{p_n}{x_n^2}:  x_1, \dots,x_m \in (0,\infty),\prod_{n=1}^mx_n \leq N  \right\} 
		\geq m \left( \prod_{n=1}^m \frac{ p_n}{N^2}\right)^{\frac{1}{m}} ,
	\end{eqnarray*}
which concludes the proof.
\end{proof}

\begin{lemma}\label{lem_local_lip}
	The following hold:
	\begin{enumerate}
		\item[(i)] For any $x,y > 0$, $|\sqrt{x}-\sqrt{y}| \leq \left(\frac{1}{\sqrt{x}}+\frac{1}{\sqrt{y}}\right)|x-y|$.
	
		\item[(ii)] Set $C> 0$. For any $x,y \in \R$, $|e^{Cx}-e^{Cy}| \leq C\left(e^{Cx}+e^{Cy}\right)|x-y|$.
	\end{enumerate}
\end{lemma}

%

\begin{lemma}\label{lem_exp_pos_rv}
For a positive random variable $X$ on $(\Omega, \cF,\mathbb{P})$, 
$\mathbb{E}[X]=\int_0^{+\infty}\mathbb{P}(X > u) du$.
\end{lemma}

\section*{Acknowledgements}
The authors would like to thank Andrea Pallavicini and Emanuela Rosazza Gianin for fruitful discussions.
The second author was supported by the Grant BIRD190200 ``Term Structure Dynamics in Interest Rate and Energy Markets: Modelling and Numerics''.


\begin{thebibliography}{50}

\bibitem{AbiJaber} Abi Jaber, E. and El Euch, O. (2019):
\emph{Multifactor approximation of rough volatility models},
SIAM Journal on Financial Mathematics, 10(2), pp. 309-349.

	
	\bibitem{Adler_1} Adler, R.J. and Taylor, J.E. (2007): \emph{Random Fields and Geometry}, Springer Monographs in Mathematics, New York, Springer-Verlag.
	

	\bibitem{Alos_T} Al\`os, E.;  Le\'on, J. A. and Vives J. (2007): \emph{On the short-time behavior of the implied volatility for jump-diffusion models with stochastic volatility}, Finance and Stochastics, 11(4), pp. 571-589.

	\bibitem{Bayer} Bayer, C.; Friz, P.K. and  Gatheral, J. (2016): \emph{Pricing under rough volatility}, Quantitative Finance, 16(6), pp. 887-904.

	\bibitem{Bayer_rates_1} Bayer, C.; Fukasawa, M. and Nakahara, S. (2022): \emph{On the weak convergence rate in the discretization of rough volatility models}, ArXiV preprint, \url{https://arxiv.org/abs/2203.02943}.
	
	\bibitem{Bayer_rates_2} Bayer, C.; Hall, E.J. and Tempone, R. (2021): \emph{Weak error rates for option pricing under linear rough volatility}, arXiv preprint, \url{https://arxiv.org/abs/2009.01219}.
	
	\bibitem{Bayer_MC} Bayer, C.; Hammouda, C.B. and Tempone, R. (2020): \emph{Hierarchical adaptive sparse grids and quasi Monte Carlo for option pricing under the rough Bergomi model}, Quantitative Finance, 20(9), pp. 1457-1473.
	
	
	\bibitem{Bennedsen} Bennedsen, M.; Lunde, A. and Pakkanen, M.S. (2017): \emph{Hybrid scheme for Brownian semistationary processes}, Finance and Stochastics, 21, pp. 931-965.
	
	\bibitem{Bergomi} Bergomi; L. (2005); \emph{Smile dynamics II}, Risk, pp. 67-73.
	
	\bibitem{Carr} Carr, P.P. and Madan, D.B. (2014): \emph{Joint modeling of VIX and SPX options at a single and common maturity with risk management applications}, IIE Transactions, 46(11), pp. 1125-1131.

	\bibitem{Chen} Chen, W.; Langrené, N.; Loeper, G. and Zhu Q. (2021): \emph{Markovian approximation of the rough Bergomi model for Monte Carlo option pricing}, Mathematics, 9(5), pp. 528. 

	\bibitem{Chow} Chow, Y.S. and Teichner E. (1997): \emph{Probability Theory}, Springer Texts in Statistics, New York,  Springer-Verlag.

	\bibitem{Corlay} Corlay, S. (2011): \emph{Quelques aspects de la quantification optimale, et applications en finance} (in English, with French summary), PhD Thesis, Universit\'e Pierre et Marie Curie.
	
	
	
	\bibitem{Fukasawa_T} Fukasawa, M. (2011): \emph{Asymptotic analysis for stochastic volatility: martingale expansion}, Finance and Stochastics, 15(4), pp. 635-654.
	
	\bibitem{Vol_is_r?} Fukasawa, M.; Takabatake, T. and  Westphal, R. (2021): \emph{Is volatility rough?},  Mathematical Finance, to appear.
	
	\bibitem{Fukasawa} Fukasawa, M. (2021): \emph{Volatility has to be rough},  Quantitative Finance, 21, pp. 1-8.
	
	\bibitem{Gassiat} Gassiat, P. (2022): \emph{Weak error rates of numerical schemes for rough volatility}, arXiv preprint, \url{https://arxiv.org/abs/2203.09298}.
		
	\bibitem{Vol_is_r} Gatheral, J.; Jaisson, T. and Rosenbaum, M. (2018): \emph{Volatility is rough}, Quantitative Finance, 18(6), pp. 933-949.
		
	\bibitem{Gatheral_Mark_Model} Gatheral, J. (2008):  \emph{ Consistent modelling of SPX and VIX options}, Presentation, Bachelier Congress, London.
	
	\bibitem{Gersho} Gersho, A. and Gray, R.M. (1992): \emph{ Vector Quantization and signal compression}, New York, Kluwer Academic Publishers.
	
	
	\bibitem{GrafLus} Graf, S. and Luschgy., H. (2007): \emph{Foundations of quantization for probability distributions}, Lecture Notes in Mathematics, 1730, Berlin Heidelberg, Springer.



\bibitem{RoughTrees} Horvath, B.; Jacquier, A. and Muguruza A. (2019): \emph{Functional central limit theorems for rough volatility}, \url{arxiv.org/abs/1711.03078}.
	
\bibitem{HJT2020} Horvath, B.; Jacquier, A. and Tankov P. (2020): \emph{Volatility options in rough volatility models}, SIAM Journal on Financial Mathematics, 11(2).
	
	\bibitem{Huh} Huh, J.; Jeon, J. and Kim, J.H. (2018): \emph{A scaled version of the double-mean-reverting model for VIX derivatives}, Mathematics and Financial Economics, 12(4), pp. 495-515.
	
	\bibitem{Jacquier} Jacquier, A.; Pakkanen, M. S. and Stone, H. (2018): \emph{Pathwise large deviations for the rough Bergomi model}, Journal of Applied Probability, 55(4), pp. 1078-1092.

	\bibitem{Muguruza} Jacquier, A.; Martini, C. and Muguruza, A. (2018): \emph{On VIX Futures in the rough Bergomi model}, Quantitative Finance, 18(1), pp. 45-61.

	\bibitem{Kallenberg} Kallenberg, O. (2002): \emph{ Foundations of Modern Probability}, 2nd edition, Probability and Its Applications, New York, Springer-Verlag.

	\bibitem{Karp} Karp, D.B. (2015): \textit{Representations and inequalities for generalized Hypergeometric functions}, J. Math. Sci. (N.Y.), 207, pp. 885-897.
	
	\bibitem{Kokholm} Kokholm, T. and Stisen, M. (2015): \emph{ Joint pricing of VIX and SPX options with stochastic volatility and jump models}, Journal of Risk Finance, 16(1), pp. 27-48.
	

	\bibitem{Luke} Luke, Y.L. (1969): \textit{The special functions and their approximations}, Volume 1, Academic Press, New York and London.

	\bibitem{Pages29} Luschgy, H. and Pag\`es, G. (2002): \emph{Functional quantization of Gaussian processes}, Journal of Functional Analysis, 196(2), pp. 486-531.
	
	\bibitem{Sharp} Luschgy, H. and Pag\`es, G. (2007): \emph{High-resolution product quantization for Gaussian processes under sup-norm distortion}, Bernoulli, 13(3), pp. 653-671.
	
	\bibitem{McCrickerd} McCrickerd, R. and Pakkanen, M.S. (2018): \emph{Turbocharging Monte Carlo pricing for the rough Bergomi model}, Quantitative Finance, 18(11), pp. 1877-1886.
	
	\bibitem{Olver} Olver, F.W.J. (1997): \emph{Asymptotics and special functions}, 2nd Edition, A.K. Peters / CRC Press.

	\bibitem{Pages} Pag\`es, G. (2007): \emph{Quadratic optimal functional quantization of stochastic processes and numerical applications}, Monte Carlo and Quasi-Monte Carlo Methods 2006, Springer, Berlin Heidelberg, pp. 101-142.
	
	\bibitem{PP2005} Pag\`es, G. and Printems, J. (2005): \emph{Functional quantization for numerics with an application to option pricing}, Monte Carlo Methods and Applications, 11(4), pp. 407-446.
	
	\bibitem{Picard} Picard, J.(2011): \emph{Representation formulae for the fractional Brownian motion}, S\'eminaire de Probabilit\'es XLIII. Lecture Notes in Mathematics, 2006, Springer-Verlag, Berlin Heidelberg, pp. 3-70.
	
	
	
	\bibitem{Sheppard} Sheppard, W.F. (1897): \emph{On the calculation of the most probable values of frequency-constants, for data arranged according to equidistant division of a scale}, Proc. Lond. Math. Soc. (3), 1(1), pp. 353-380.

	\bibitem{Steele} Steele, J. M. (2004): \emph{The Cauchy-Schwarz Master-Class}, Cambridge University Press.
	

\end{thebibliography}
\end{document}